\documentclass[lettersize,journal]{IEEEtran}

\usepackage{cite}
\usepackage{amsmath,amssymb,amsfonts}
\usepackage{graphicx}
\usepackage{textcomp}
\usepackage{xcolor}
\usepackage[version=4]{mhchem}
\usepackage{siunitx}
\usepackage{array}
\usepackage{longtable}
\usepackage{nomencl}
\usepackage{mathrsfs} 
\usepackage{mathtools} 
\usepackage{optidef}
\usepackage{multicol}
\usepackage{svg}
\usepackage{lipsum}
\usepackage{calrsfs}
\usepackage{scalerel}
\usepackage{stackengine}
\usepackage{booktabs}
\usepackage{caption}
\usepackage{enumerate}
\usepackage[english]{babel}
\usepackage{amsthm}
\usepackage{amsmath}
\usepackage{boxedminipage}
\usepackage{multirow,nicefrac}
\usepackage{makecell,upgreek}
\usepackage{footnote}
\usepackage{longtable}
\usepackage[shortlabels]{enumitem}
\usepackage{tablefootnote}
\usepackage[T1]{fontenc}
\usepackage{verbatim} 
\usepackage{float}
\usepackage[utf8]{inputenc}
\usepackage{url}
\usepackage{cite}
\usepackage[breaklinks=true]{hyperref}
\usepackage{amsfonts,amssymb,mathrsfs}
\usepackage{mathtools}
\usepackage{latexsym}
\usepackage{relsize}
\usepackage{thm-restate}
\usepackage{dsfont}
\usepackage{fullpage}
\usepackage[noend]{algpseudocode}
\usepackage{algorithm}
\usepackage{subcaption}

\usepackage[normalem]{ulem}
\usepackage{cancel}

\newcommand{\be}[0]{\begin{equation}}
\newcommand{\ben}[0]{\begin{equation*}}
\newcommand{\een}[0]{\end{equation*}}
\newcommand{\bena}[0]{\begin{eqnarray*}}
\newcommand{\eena}[0]{\end{eqnarray*}}
\newcommand{\bea}[0]{\begin{eqnarray}}
\newcommand{\eea}[0]{\end{eqnarray}}

\DeclareMathAlphabet{\mathcal}{OMS}{cmsy}{m}{n}
\DeclareMathSymbol{\shortminus}{\mathbin}{AMSa}{"39}

\usepackage{prettyref}

\usepackage[capitalise,nameinlink]{cleveref}
\Crefname{equation}{Eq.}{Eqs.}
\Crefname{assumption}{Assumption}{Assumptions}
\Crefname{condition}{Condition}{Conditions}
\Crefname{claim}{Claim}{Claims}

\usepackage{breakcites,bm}

\hypersetup{colorlinks,citecolor=blue,linkcolor = blue}

\numberwithin{equation}{section}

\DeclareFontFamily{U}{mathx}{\hyphenchar\font45}
\DeclareFontShape{U}{mathx}{m}{n}{
      <5> <6> <7> <8> <9> <10>
      <10.95> <12> <14.4> <17.28> <20.74> <24.88>
      mathx10
      }{}
\DeclareSymbolFont{mathx}{U}{mathx}{m}{n}
\DeclareFontSubstitution{U}{mathx}{m}{n}
\DeclareMathAccent{\widecheck}{0}{mathx}{"71}
\DeclareMathAccent{\wideparen}{0}{mathx}{"75}

\newcommand{\ignore}[1]{}

	\theoremstyle{plain}
	\newtheorem{theorem}{Theorem}

	\newtheorem{lemma}{Lemma}[section]
	
	\newtheorem{corollary}{Corollary}[section]
	\newtheorem{proposition}[lemma]{Proposition}

	\theoremstyle{definition}

	\newtheorem{definition}{Definition}[section]

	\newtheorem{remark}{Remark}[section]

  \newtheorem{assumption}{Assumption}[section]
  \newtheorem{condition}{Condition}[section]

\makeatletter
\newcommand{\neutralize}[1]{\expandafter\let\csname c@#1\endcsname\count@}
\makeatother

\newtheorem*{theorem*}{Theorem}
\newtheorem*{lemma*}{Lemma}
\newtheorem*{corollary*}{Corollary}
\newtheorem*{proposition*}{Proposition}
\newtheorem*{claim*}{Claim}
\newtheorem*{fact*}{Fact}
\newtheorem*{observation*}{Observation}

\newtheorem*{definition*}{Definition}
\newtheorem*{remark*}{Remark}
\newtheorem*{example*}{Example}

\newtheoremstyle{named}{}{}{\itshape}{}{\bfseries}{}{.5em}{\Cref{#3} {\normalfont (informal)} }
{}
\theoremstyle{named}

\theoremstyle{plain}

\DeclareMathAlphabet{\mathbfsf}{\encodingdefault}{\sfdefault}{bx}{n}

\DeclareMathOperator*{\argmin}{arg\,min}

\renewcommand{\leq}{~\le~}
\renewcommand{\geq}{~\ge~}

\let\oldtfrac\tfrac
\renewcommand{\tfrac}[2]{\smash{\oldtfrac{#1}{#2}}}

\let\nablaold\nabla
\renewcommand{\nabla}{\nablaold\mkern-2.5mu}

\newcommand\redsout{\bgroup\markoverwith{\textcolor{red}{\rule[0.5ex]{2pt}{0.4pt}}}\ULon}

\begin{document}


\title{\, \\ Flight Envelope Protection for a Hypersonic Glide Vehicle Using Adaptive Safety-Critical Control}


\author{Johannes~Autenrieb,~\IEEEmembership{Member,~IEEE,} Peter~A.~Fisher,~\IEEEmembership{Member,~IEEE,} and Anuradha~M.~Annaswamy,~\IEEEmembership{Fellow,~IEEE}
\thanks{J. Autenrieb is with the German Aerospace Center (DLR), Institute of Flight Systems, 38108, Braunschweig, Germany. E-mail: {\tt johannes.autenrieb@dlr.de}.}
\thanks{P.A. Fisher and A.M. Annaswamy are with the Department of Mechanical Engineering, Massachusetts Institute of Technology, Cambridge, MA 02139, USA. E-mail: {\tt\{pafisher, aanna\}@mit.edu}.
\newline The second author would like to acknowledge the support of the Boeing Strategic University Initiative and the Air Force Research Laboratory.}
}


\hyphenation{op-tical net-works semi-conduc-tor IEEE-Xplore}

%

\maketitle


\begin{abstract}
This paper presents an adaptive safety-critical control framework for flight envelope protection (FEP) of hypersonic glide vehicles (HGVs) under model uncertainty. The proposed architecture treats input and state constraints through two complementary but distinct mechanisms. At the input end, an adaptive controller with a calibrated closed-loop reference model (CCRM) is designed to accommodate magnitude-limited control inputs and the effects of actuator saturation. Building on this input-constrained adaptive control architecture, flight envelope state constraints are enforced at the state end by an error-based safety filter (EBSF) based on control barrier functions (CBFs). The EBSF modifies the reference command and adapts the admissible safe set online using the measured mismatch between the reference model and the uncertain plant, thereby preserving forward invariance during transient adaptation. Simulation results for the DLR generic hypersonic glide vehicle 2 (GHGV-2) demonstrate stable, high-performance tracking under magnitude-limited control inputs while maintaining flight envelope constraints in the presence of model uncertainty.
 
\end{abstract}


%
\IEEEpeerreviewmaketitle


\section{Introduction}
\label{Introduction}
\IEEEPARstart{H}{ypersonic} glide vehicles (HGVs) operate in regimes characterized by strong aerodynamic coupling, rapidly varying flow conditions, severe thermal and structural loads, and limited control authority. Integrated flight control systems must therefore not only achieve high tracking performance, but also respect stringent operational constraints. In this context, flight envelope protection (FEP) is essential to ensure that quantities such as angle of attack, load factor, and altitude remain within admissible limits, thereby preventing loss of control, structural overstress, excessive heating, or violation of aerodynamic controllability constraints~\cite{Oudin2017}.

Classical FEP methods are often implemented through reference clipping or model-based command filtering~\cite{Tang2009,Falkena2010}. While attractive due to their simplicity and compatibility with existing flight control architectures, these methods provide limited guarantees. Reference clipping does not account for closed-loop dynamics and may therefore lead to transient constraint violations~\cite{Seo2017}. Model-based command filters improve transient behavior by incorporating simplified dynamics~\cite{Lombaerts2017}, but remain sensitive to model mismatch and may become unreliable under significant uncertainty. This is particularly critical for hypersonic systems, where aerodynamic uncertainty is unavoidable and even short constraint violations may compromise safety.

Control barrier functions (CBFs) offer a principled alternative by formulating FEP as a set-invariance problem. Safety is encoded through a continuously differentiable function whose superlevel set represents the admissible operating region, and control inputs are chosen such that this set remains forward invariant~\cite{Ames_2017}. This makes CBFs well suited for real-time optimization-based safety-critical control~\cite{Agrawal2021-aw}. However, standard CBF formulations are typically derived for known dynamics. When applied directly to uncertain systems, the true plant trajectory may deviate from the nominal model used in the barrier condition, and forward invariance is no longer guaranteed.

Adaptive control addresses parametric uncertainty by estimating unknown model parameters online and adjusting the control law accordingly~\cite{Narendra2005}. In particular, model reference adaptive control (MRAC) and related methods provide well-established stability and tracking guarantees for uncertain systems. However, adaptive control alone does not ensure constraint satisfaction during the transient learning phase. Since parameter estimates may temporarily differ significantly from their true values, plant-reference mismatch can cause constraint violations before convergence. Thus, stability and tracking guarantees are insufficient for safety-critical systems in which constraints must hold for all time.

Several adaptive CBF methods have been proposed to address safety under parametric uncertainty. In~\cite{Ames_2020a}, adaptive CBF conditions are obtained by embedding uncertain parameters directly into the barrier constraints and enforcing the resulting inequalities over admissible parameter and adaptation-gain sets. Extensions accounting for input-matrix uncertainty have also been considered~\cite{Wang2024}. While these approaches provide important theoretical guarantees, they typically formulate safety directly at the plant input-level and can lead to restrictive invariance conditions. Robust adaptive CBFs reduce some of the conservatism~\cite{Lopez2021}, but their integration with classical MRAC structures remains challenging.

A reference-based safety architecture was introduced in~\cite{Autenrieb2023}, where a CBF safety filter certifies a safe reference command for a considered reference model, while an adaptive controller tracks this certified reference with the uncertain plant. This approach was extended in~\cite{fisher_2025} using an error-based safety filter (EBSF), where the measured mismatch between the reference model and uncertain plant is incorporated directly into the safety condition. This allows the admissible set for the reference dynamics to be adjusted online, reducing conservatism compared with static safety buffers and, more importantly, ensuring forward invariance of the closed-loop adaptive system.

This paper develops an adaptive safety-critical control framework for FEP of the DLR generic hypersonic glide vehicle 2 (GHGV-2) that builds on the approaches in \cite{Autenrieb2023, fisher_2025}. The framework integrates the treatment of input and state constraints within a unified adaptive control architecture. Input constraints caused by actuator magnitude limits and overactuation are handled using an adaptive controller with a calibrated closed-loop reference model (CCRM), enabling systematic integration with control allocation. State constraints associated with the flight envelope are enforced using a CBF-based EBSF, which tightens the safety constraints online based on measured mismatch between the plant and reference model. This separation is particularly suitable for overactuated hypersonic vehicles, where actuator limits and flight envelope constraints must be enforced simultaneously under significant uncertainty.

The main contributions of this paper are as follows. First, FEP for the GHGV-2 is formulated as a state-constrained adaptive control problem under parametric uncertainty and actuator limitations. The considered envelope constraints include load-factor limits associated with stall and structural boundaries, as well as operational Mach-altitude corridor constraints, which are translated into admissible angle-of-attack bounds. Second, a CCRM-based adaptive control architecture is developed to accommodate magnitude-limited control inputs and enable integration with control allocation for the overactuated GHGV-2. Third, a CBF-based EBSF is incorporated to enforce flight envelope state constraints during adaptation. Unlike static-buffer approaches, the EBSF uses measured tracking error to adjust the admissible safe set online, reducing conservatism while preserving forward invariance. Finally, nonlinear simulations of the DLR GHGV-2 are carried out to
demonstrate the advantages of the proposed adaptive control architecture.


The remainder of this paper is organized as follows. Section~\ref{Flight_Dynamics} introduces the GHGV-2 flight dynamics model, the linearized design model, and the attainable moment set used to represent actuator limitations. Section~\ref{Problem Formulation} formulates the FEP problem under uncertainty by deriving load-factor and aerodynamic-corridor constraints and expressing them as state constraints. Section~\ref{Safe_AC} presents the proposed adaptive safety-critical control architecture, including the CCRM-based adaptive controller for input-constrained tracking and the EBSF for state-constraint enforcement. Section~\ref{Numerical_Assessment} presents nonlinear GHGV-2 simulation results and comparisons with baseline and other safety-critical architectures. Section~\ref{Conclusions} concludes the paper.

\section{The Hypersonic Glide Vehicle's Flight Dynamics Model}
\label{Flight_Dynamics}
This paper focuses on the GHGV-2 designed to achieve enhanced lift-to-drag ratios during sustained flight in high Mach number regimes~\cite{Gruhn2020}.  
A typical hypersonic flight starts with a launch using a booster rocket which transports the vehicle to an altitude of approximately 100km~\cite{Autenrieb2021,Autenrieb2024}. The vehicle then separates and transitions to a free flight phase. The onboard GNC systems manage this transition and guide the vehicle along a prescribed trajectory using propulsion-system-based corrections for high-altitude exoatmospheric flight. Subsequently, when the vehicle reenters denser atmospheric layers, it transitions from thruster-based to aerodynamic control and initiates the endoatmospheric glide phase. This paper focuses on this glide phase where the goal is to follow a position trajectory towards a designated terminal location. 

We begin with a full nonlinear rigid-body six-degree-of-freedom model of the GHGV-2 (see \cite{Autenrieb2021, Gruhn2018,Autenrieb2023sc,Autenrieb2024} for details). The complete model describes both the translational motion of the vehicle relative to the Earth-fixed frame and the rotational motion about the body-fixed axes. Consequently, the overall system dynamics encompasses the position, velocity, and flight path dynamics required for trajectory propagation, as well as the aerodynamic attitude dynamics governing the evolution of the vehicle orientation.

The relevant dynamics can be defined as:
\begin{equation}
    \label{general_nonlinaer_dynamics_1}
    {\dot{X}}(t) = f({X}(t),\, { \delta}(t), {\sigma}),
\end{equation}
where $\sigma$ denotes the operating point of the GHGV-2 and is given by
\begin{equation}
    \label{eqn:flight point params}
    {\sigma} = \begin{bmatrix} \operatorname{Ma} & H \end{bmatrix}^T,
\end{equation}
where $\operatorname{Ma}$ is the Mach number and $H$ is the altitude, ${X}(t) \in \mathbb{R}^n$  is the system state, given by
\begin{equation}
    \label{eqn:state vector definition}
    {X}(t) = 
    \begin{bmatrix}
     \mu(t) & 
    \alpha(t) & 
    \beta(t) & 
     p(t) & 
     q(t) & 
     r(t)\end{bmatrix}^T, 
\end{equation}
$\mu$ is the  flight path bank angle, $\alpha$ the angle of attack, $\beta$ the sideslip angle, and $p$, $q$, and $r$ are the body-fixed rotational rates. The control input ${\delta}(t) \in \mathbb{R}^o$ is given by
\begin{equation}
    \label{eqn:input vector definition evaluation model}
     {\delta}(t) = \begin{bmatrix} 
    \delta_1(t) & 
    \delta_2(t) & 
    \delta_3(t) & 
    \delta_4(t) \end{bmatrix}^T,
\end{equation}
which correspond to the surface deflections (see Fig.~\ref{fig:Available control effectors during endoatmospheric operations}) of the upper left fin $\delta_{1}(t)$, upper right fin $\delta_{2}(t)$, lower left fin $\delta_{3}(t)$, and lower right fin $\delta_{4}(t)$, respectively.

The tracking goal of interest is to design $\delta(t)$ such that $X(t)$ follows a prescribed trajectory $X^*(t)$ in the glide phase. We assume that additional guidance systems are  in place to ensure that $X^*(t)$  is such that the operating point $\sigma $ is at  a desired location. 

\begin{figure}
\centering
\includegraphics[width=\columnwidth]{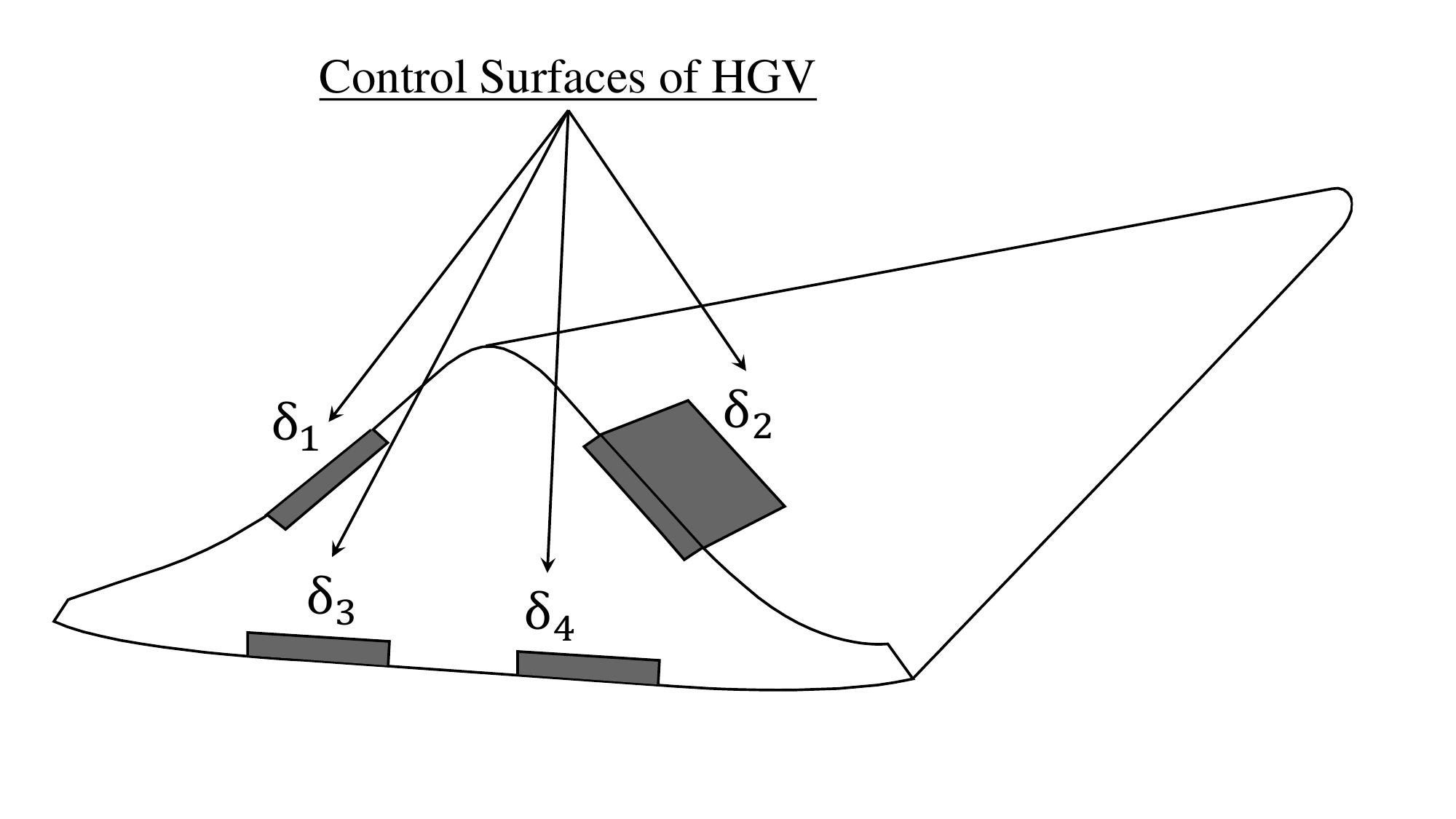}
\caption{Rear view of the GHGV-2 and available control effectors during endoatmospheric operations \cite{Autenrieb_2025b}.}
\label{fig:Available control effectors during endoatmospheric operations}
\end{figure}

\subsection{Design Model for Flight Controller Synthesis}
\label{sec_Design_Model}
In order to facilitate a control design, we develop the following design model, obtained by linearizing \eqref{general_nonlinaer_dynamics_1} around a trim state $X_0$ and trim input $U_0$ to obtain
\begin{equation}
\label{eqn:State-Space System 1}
{\dot{x}}(t) = A({\sigma}_0)\,{x}(t)+B({\sigma}_0)\,{u}(t),
\end{equation}
where $x(t)=X(t)-X_0$, $u(t)=U(t)-U_0$, $f(X_0,\delta_0,\sigma_0)=0$.  $(X_0,U_0)$ corresponds to the operating point  $\sigma_0= \begin{bmatrix} \operatorname{Ma}_0 & H_0 \end{bmatrix}^T$. The control input $u(t)$ is given by
\begin{equation}
    \label{eqn:input vector definition}
     {u}(t)= \begin{bmatrix} 
    M_x(t) & 
    M_y(t) & 
    M_z(t)\end{bmatrix}^T,
\end{equation}
which correspond to the moments around the $x,y,z$ axes, respectively, and can be viewed as a virtual control input~\cite{Johnson2006, Chen2021}.  The tracking goal can be restated as using the linear model in \eqref{eqn:State-Space System 1} for the design of $u(t)$ in \eqref{eqn:input vector definition} so that $x(t)$ tracks a command $x^*$. 

In order to map the desired virtual control input $u(t)$ to physical control surface deflection $\delta^*(t) = \delta(t) - \delta_0$, where $\delta^*(t)$ denotes the effective deflection relative to the trim condition $\delta_0$, a control allocation algorithm is used~\cite{Autenrieb2026JGCD}. We define an admissible set $\Phi$ as
\begin{equation}
\label{eqn:admissible _input_set}
\begin{split}
{\Phi} := \{\delta^*(t) \in \mathbb{R}^{o} \, \rvert \, \forall i \in [1,o] : &\delta^*_{min,i}\\ &\leq \delta^*_i(t)\leq \delta^*_{max,i}\},
\end{split}
\end{equation} 
where $\delta^*_{min,i}$ and $\delta^*_{max,i}$ are allowable deflection limits for each control surface $i$. These are also the limits the adaptive safety-critical control system has to respect while ensuring desired tracking and guaranteed FEP.


Using the admissible set $\Phi$, we construct a physical control volume $\mathcal{U}$ as
\begin{equation}
\label{eqn:control volume moment}
    \mathcal{U} := \{  u(t) \in \mathbb{R}^m \, | \,   u(t) = B_\delta({\sigma}_0) \delta^*(t) \wedge \delta^*(t) \in {\Phi} \},
\end{equation}
with $B_\delta(\sigma_0)$ being the control effectiveness matrix defined by the Jacobian
\begin{equation}
    \label{eqn:Control effectivnes matrix}
    B_\delta({\sigma}_0) = \begin{bmatrix} 
    \displaystyle \frac{\partial M_x}{\partial \delta_1 }&\displaystyle\frac{\partial M_x}{\partial \delta_2 }&\displaystyle\frac{\partial M_x}{\partial \delta_3 }&\displaystyle\frac{\partial M_x}{\partial \delta_4 }\\[8pt]
    \displaystyle\frac{\partial M_y}{\partial \delta_1 }&\displaystyle\frac{\partial M_y}{\partial \delta_2 }&\displaystyle\frac{\partial M_y}{\partial \delta_3 }&\displaystyle\frac{\partial M_y}{\partial \delta_4 }\\[8pt]
    \displaystyle\frac{\partial M_z}{\partial \delta_1 }&\displaystyle\frac{\partial M_z}{\partial \delta_2 }&\displaystyle\frac{\partial M_z}{\partial \delta_3 }&\displaystyle\frac{\partial M_z}{\partial \delta_4 }
    \end{bmatrix}.
\end{equation}\medskip

This control volume $\mathcal{U}$ essentially defines the limits on the control moments that can be generated by the allowable deflection limits as in \eqref{eqn:admissible _input_set}. This in turn allows us to construct a saturation function $R_u( u(t))$ defined as 
\begin{equation}
R_u( u(t)) =
\begin{cases}
 u(t), & \text{if } \| u(t)\| \leq g( u(t)), \\
\bar{ u}(t), & \text{if } \| u(t)\| > g( u(t)),
\end{cases}
\label{eqn:MRAC_SettlementFunction}
\end{equation}
where \(g( u(t))\) describes the maximum permissible value of the moment vector,  which is chosen as~\cite{Gaudio_2019b}:
\begin{equation}
g( u(t)) = \left( \sum_{i=1}^{3} \left( \frac{\hat{ e}_i}{u_{\max,i}} \right)^2 \right)^{-1/2},\label{eqn:g(u)}
\end{equation}
that is, an ellipsoid, where \(\hat{ e} = \dfrac{ u(t)}{\| u(t)\|}\) describes the normalized direction vector, \(u_{\text{max},i}\) represents the magnitude of $u(t)$ in direction $i$ of the ellipsoid along the three semi-axes and \(\bar{ u}(t) = \hat{ e} \cdot g( u(t))\) is the saturated input vector projected onto the edge of the ellipsoid (see Fig. ~\ref{fig:AMS_Ellipsoid}). The three magnitudes \(u_{\text{max},i}\) in $g( u(t))$, $i=1,2,3$ are chosen such that the corresponding admissible set $\Phi$ in \eqref{eqn:admissible _input_set} for a choice of $u(t)$ as in \eqref{eqn:control volume moment} inscribes the largest ellipsoid spanned by $g(u(t))$. 
This construction preserves the direction of the original control output while ensuring that the commanded moment remains inside the admissible set, thereby providing a computationally efficient mechanism for incorporating magnitude limits on control inputs into an adaptive controller.
\begin{figure*}[h!]
\centering
\includegraphics[width=0.7\textwidth]{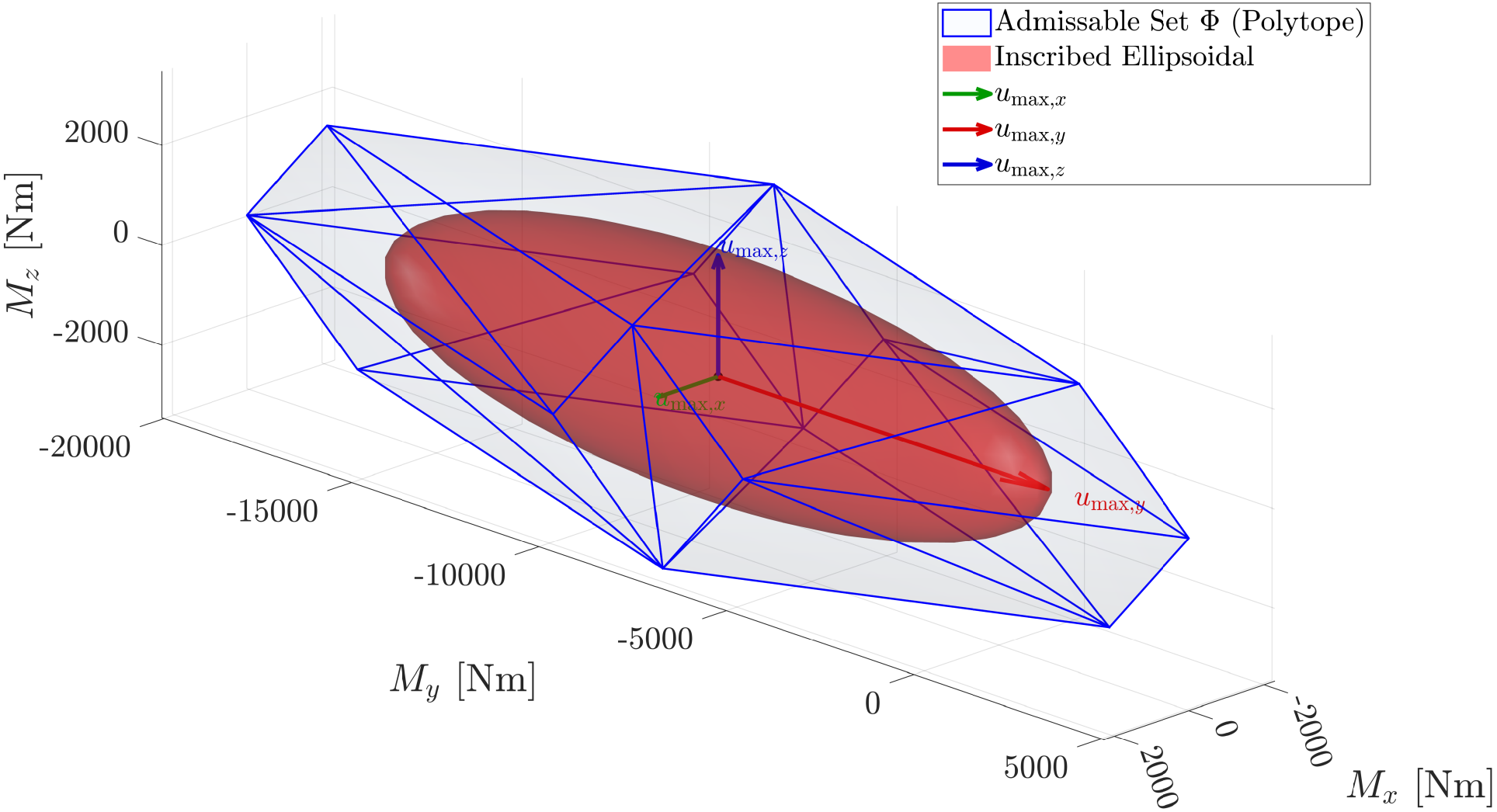}
\caption{Attainable Moment Set with inscribed ellipsoid.}
\label{fig:AMS_Ellipsoid}
\end{figure*}


\subsection{Uncertainty-Driven Model Modifications}

In addition to incorporating magnitude limits on the control input, our vehicle design model is also assumed to have parametric uncertainties in aerodynamics and inaccuracies in control effectiveness.
These parametric uncertainties lead to an uncertain design model of the plant in the form
\begin{equation}
\label{eqn:LinearPlantModel_Problem}
\dot{x}(t) = A(\sigma_0)x(t) + B(\sigma_0)\Lambda\, R_u(u(t)),
\end{equation}
where $A(\sigma_0) \in \mathbb{R}^{n \times n}$ and $\Lambda \in \mathbb{R}^{m \times m}$ are unknown, with $\Lambda$ being diagonal and strictly positive, and $B(\sigma_0) \in \mathbb{R}^{n \times m}$ is assumed to be known.

\section{The Flight Envelope Protection Problem in the Presence of Uncertainties}
\label{Problem Formulation}
As discussed in Section~\ref{Flight_Dynamics}, the GHGV-2 operates within a hierarchical guidance and control architecture in which higher-level trajectory generation and flight-path guidance modules determine the desired vehicle motion and provide required attitude states to the inner-loop flight control system. Consequently, the considered state vector $X(t)$ must evolve such that the vehicle accurately follows $X^*(t)$.
In addition to achieving satisfactory tracking performance of the guidance-generated attitude commands, it is also important that the hypersonic glide vehicle operates within physically admissible conditions that prevent excessive thermal loads, structural overstressing, actuator saturation, and loss of aerodynamic controllability. Violations of these limits may compromise both vehicle safety and mission success. Consequently, the flight control system must not only stabilize the vehicle and track the commanded motion, but also actively enforce operational constraints throughout the maneuver. In other words, we need to integrate a flight envelope protection framework into the problem that will guarantee safe operation even while achieving a desired closed-loop tracking performance.


The specific constraints considered in this work are associated with stall and structural load limits, as well as aerodynamic corridor requirements arising from operational boundaries. 
In what follows, we will model these constraints as limits on relevant states. This in turn will allow us to pose an overall tracking problem of a linear dynamic system subject to state constraints and parametric uncertainties.


\subsection{Load Constraints Due to Stall and Structural Limits}
\label{sec:Stall_Structural_Limits}

As the vehicle moves through the atmosphere, it generates aerodynamic lift, which scales with the dynamic pressure $Q(t)$, defined as
\begin{equation}
Q(t) = \frac{1}{2}\rho(t)V(t)^2,
\end{equation}
where $\rho(t)$ denotes the air density and $V(t)$ the flight speed. Since $V(t)$ is determined by the Mach number and the altitude, and since $\rho(t)$ is likewise a function of the altitude, the dynamic pressure $Q_0$ can be assumed to be pointwise constant for a given operating point $\sigma_0$.

The structural loads induced by aerodynamic forces are commonly expressed in terms of the normal load factor $n_z(t)$, which can be written as
\begin{equation}
\begin{split}
n_z(t) = \frac{Q_0 S}{mg}
\big(
& C_{Z,\alpha}(\sigma_0)\alpha(t)
+ C_{Z,\delta_1}(\sigma_0)\delta_1(t)\\
&+ C_{Z,\delta_2}(\sigma_0)\delta_2(t) 
 + C_{Z,\delta_3}(\sigma_0)\delta_3(t)\\
&+ C_{Z,\delta_4}(\sigma_0)\delta_4(t)
\big),
\end{split}
\label{eqn:nz_realationship}
\end{equation}
where $S$ denotes the aerodynamic reference area, $m$ the vehicle mass, $g$ the gravitational acceleration, and the aerodynamic coefficients depend on the operating point $\sigma_0$.
\begin{figure*}[h!]
\centering
\includegraphics[width=0.8\textwidth]{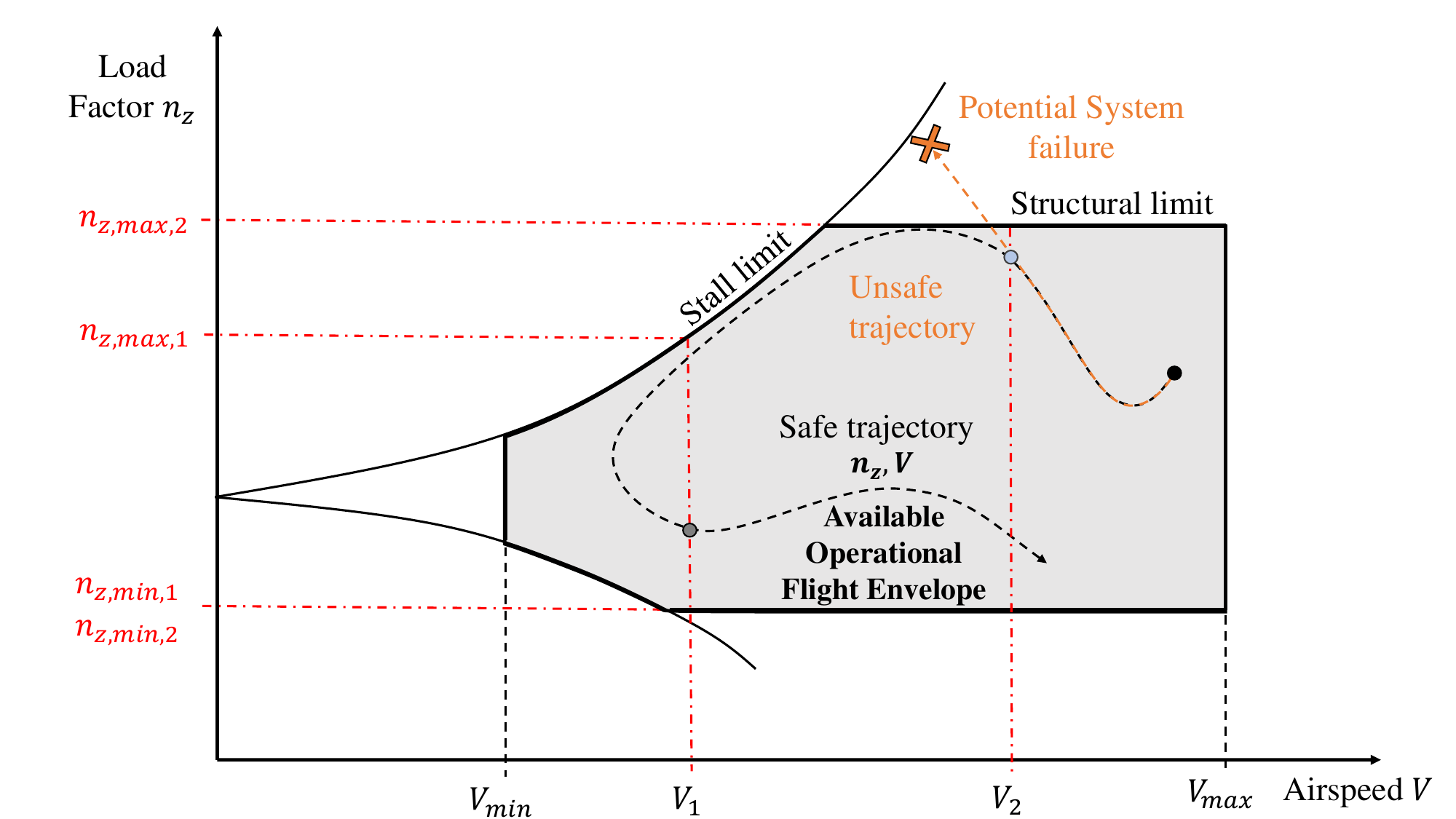}
\caption{Example of a \(V\)–\(n\) diagram illustrating the relationship between load factors and airspeed, with limits defined by stall and structural loads.}
\label{fig:V-n_Diagram}
\end{figure*}


To ensure safe operation, the vehicle must remain within a bounded set of admissible aerodynamic conditions, commonly represented by the \(V\)--\(n\) diagram shown in Fig.~\ref{fig:V-n_Diagram}. This diagram defines an upper load boundary \(n_{z,\max}(V)\), associated with structural load limitations, and a lower load boundary \(n_{z,\min}(V)\), associated with stall-related limitations. Since the present work focuses on the glide phase at a fixed operating point, the flight speed is assumed to be prescribed by the chosen glide trajectory. Consequently, the velocity-dependent load bounds can be evaluated pointwise at any given \(V_i\). Since the flight speed \(V_i\) is determined by the Mach number and altitude, the velocity dependence can be reinterpreted as a dependence on the operating point \(\sigma\), and evaluated pointwise at \(\sigma_i\), yielding corresponding constant admissible limits \(n_{z,\min}(\sigma_i)\) and \(n_{z,\max}(\sigma_i)\) (see Fig.~\ref{fig:V-n_Diagram}). By repeating this process for all operating points, we can determine the total admissible set of load factors as
\begin{equation*}
\begin{split}
S_{n_z}
:=
\{
n_z(t)\in\mathbb{R}
\;|\;
&n_{z,\min}(\sigma_0)\\
&\leq n_z(t)
\leq n_{z,\max}(\sigma_0)
\},
\end{split}
\end{equation*}
for all $t \geq t_0$.

To enforce these constraints within a control-theoretic framework, the admissible set must be reformulated in terms of the relevant system states. For this purpose, we return to \eqref{eqn:nz_realationship} and note that the contribution of control surface deflections to the load factor is comparatively small compared with that of $\alpha$, as control surfaces are primarily designed to generate moments rather than lift forces. This allows us to bound $n_z$ as 
\begin{equation}
\underline{n}_z
\leq
n_z(t)
\leq
\overline{n}_z,
\label{eqn:nz_realationship_simplified}
\end{equation}
where
\begin{equation*}
\underline{n}_z
=
\frac{Q_0 S}{mg}
\, C_{Z,\alpha}(\sigma_0)
\,\alpha^{load}_{\min}(\sigma_0),
\end{equation*}
\begin{equation*}
\overline{n}_z
=
\frac{Q_0 S}{mg}
\, C_{Z,\alpha}(\sigma_0)
\,\alpha^{load}_{\max}(\sigma_0).
\end{equation*}
By inverting \eqref{eqn:nz_realationship_simplified}, the load constraints can equivalently be expressed as the following operating-point–dependent admissible set,
\begin{equation}
\begin{split}
S^{load}_{\alpha}(\sigma_0)
:=
\{
\alpha(t)\in\mathbb{R}
\;|\;
&\alpha^{load}_{\min}(\sigma_0)\\
&\leq \alpha(t)
\leq \alpha^{load}_{\max}(\sigma_0)
\},
\end{split}
\label{eqn:alpha_crit_load}
\end{equation}
for all $t \geq t_0$. 

\subsection{Aerodynamic Corridor Constraints} \label{sec:Aerodynamic_Corridor_Constraints}
In addition to stall and structural limitations, safe operation of hypersonic glide vehicles requires adherence to an admissible aerodynamic flight corridor in the Mach-altitude domain. Deviations from this corridor may lead to insufficient aerodynamic control authority at high altitudes or excessive thermal and structural loads at lower altitudes~\cite{Li_2017,Takahashi2023-gp}.  Often an offline optimization process is carried out to determine this corridor (see 
\begin{figure}
\centering
\includegraphics[width=\columnwidth]{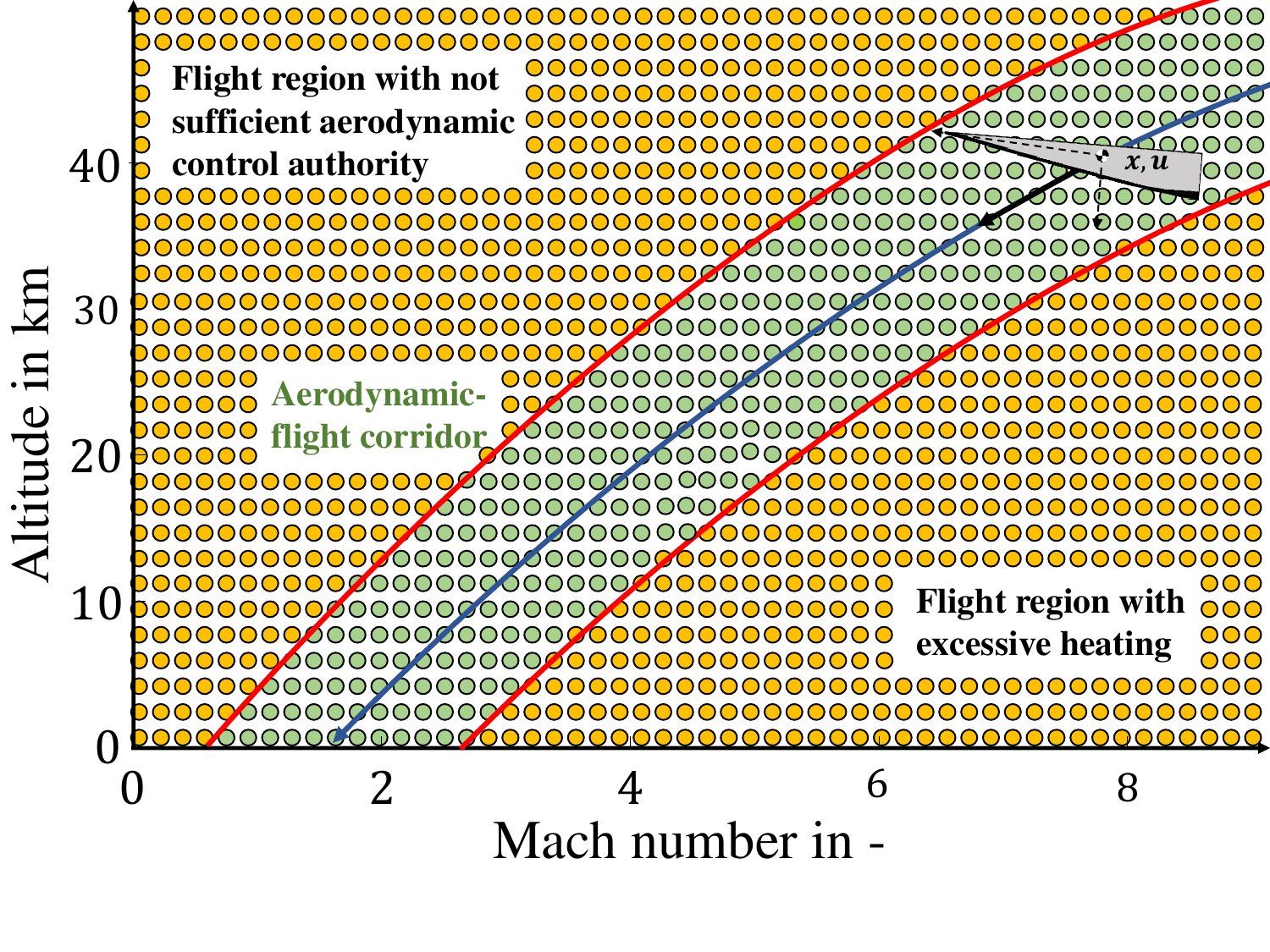}
\caption{Conceptual illustration of an operational flight corridor and reference trajectory for a hypersonic glide vehicle in the Mach-altitude space.}
\label{fig:Flight_Envelope_Operating_Points}
\end{figure}
Fig.~\ref{fig:Flight_Envelope_Operating_Points}), where the upper boundary is defined by a potential loss of control authority that a vehicle would experience due to low atmospheric density and the lower boundary is determined by potentially severe aerodynamic heating and potential structural degradation of the hypersonic system~\cite{An_2020}. For the GHGV-2, the corridor is determined via an offline process, which as a result delivers a set of allowable Mach number-altitude combinations. One could then view the blue curve in the figure as belonging to the safe region. In what follows, we translate the operation of the glide phase within this safe corridor 
\begin{equation} 
H_{\min}(\operatorname{Ma}) \leq H(t) \leq H_{\max}(\operatorname{Ma}),
\label{eq:corridor_bounds} 
\end{equation}
into state constraints. As before, if we assume that the flight speed is prescribed, the bounds $H_{\min}$ and $H_{\max}$ can be assumed to be known constants, which are obtained offline through trajectory optimization and aerodynamic–thermal feasibility analysis. Using kinematic relations, constant speed assumption, and flight path angle dynamics for the glide phase at small bank angles, we can obtain the relation

\begin{equation} 
\ddot{H}(t) = \frac{\cos\gamma(t)}{m} 
\left(\bar{L}(t)-mg\cos\gamma(t)\right), \label{eq:H_ddot_lift}
\end{equation}
which reveals that the vertical acceleration of the vehicle is directly governed by the lift force. In particular, limiting the achievable lift limits the magnitude of $\ddot{H}(t)$ and therefore constrains how rapidly the altitude can increase or decrease.
In the vicinity of the operating point $\sigma_0$, the aerodynamic lift force can be expressed as an affine function of the angle of attack, 
\begin{equation} \bar{L}(t) = Q_0 S\left(C_{\bar{L},0}(\sigma_0)+C_{\bar{L},\alpha}(\sigma_0)\alpha(t)\right),
\label{eq:lift_alpha_corridor}
\end{equation}
where $Q_0$ denotes the locally constant dynamic pressure and $C_{\bar{L},0}(\sigma_0)$,$C_{\bar{L},\alpha}(\sigma_0)$ are the lift coefficients evaluated at the operating point.
Substituting \eqref{eq:lift_alpha_corridor} into \eqref{eq:H_ddot_lift} yields
\begin{equation}
\begin{split}\ddot{H}(t)=&\frac{\cos\gamma(t)}{m}\Big(Q_0 S C_{\bar{L},0}(\sigma_0)-mg\cos\gamma(t)\Big) \\&+\frac{Q_0 S C_{\bar{L},\alpha}(\sigma_0)}{m}\cos\gamma(t)\,\alpha(t).
\end{split}
\label{eq:H_ddot_alpha}
\end{equation}
That is, the corridor constraint in \eqref{eq:corridor_bounds} can be represented by an operating-point–dependent admissible set for $\alpha(t)$ of the form
\begin{equation}
\begin{split}
S_{\alpha}^{\mathrm{corr}}(\sigma_0)
:=
\{
\alpha(t)\in\mathbb{R}
\;|\;
&
\alpha_{\min}^{\mathrm{corr}}(\sigma_0) \\
&
\leq
\alpha(t)
\leq
\alpha_{\max}^{\mathrm{corr}}(\sigma_0)
\},
\end{split}
\label{eq:alpha_corridor_constraints}
\end{equation}
for all $t \geq t_0$, where the bounds $\alpha_{\min}^{\mathrm{corr}}(\sigma_0)$ and
$\alpha_{\max}^{\mathrm{corr}}(\sigma_0)$
are chosen such that the resulting altitude evolution remains within the prescribed corridor.

\subsection{Unified Constraint Representation and Critical Constraint Selection}
\label{sec:Unified_AoA_Constraints}

Combining the stall and structural load considerations as in \eqref{eqn:alpha_crit_load}, and aerodynamic corridor requirements 
as in \eqref{eq:alpha_corridor_constraints}, we define an overall 
safe region
\begin{equation}
S_{\alpha}(\sigma_0)
=
S_{\alpha}^{\mathrm{load}}(\sigma_0)
\;\cap\;
S_{\alpha}^{\mathrm{corr}}(\sigma_0),
\label{eq:S_alpha_intersection}
\end{equation}
which corresponds to the state constraints
\begin{equation}
\alpha_{\mathrm{crit,min}}(\sigma_0)
\leq
\alpha(t)
\leq
\alpha_{\mathrm{crit,max}}(\sigma_0).
\label{eq:alpha_crit_interval}
\end{equation}
where
\begin{equation*}
\alpha_{\mathrm{crit,max}}(\sigma_0)
=
\min\!\left(
\alpha_{\max}^{\mathrm{load}}(\sigma_0),
\alpha_{\max}^{\mathrm{corr}}(\sigma_0)
\right),
\label{eq:alpha_crit_max}
\end{equation*}
\begin{equation*}
\alpha_{\mathrm{crit,min}}(\sigma_0)
=
\max\!\left(
\alpha_{\min}^{\mathrm{load}}(\sigma_0),
\alpha_{\min}^{\mathrm{corr}}(\sigma_0)
\right).
\label{eq:alpha_crit_min}
\end{equation*}

\subsection{Overall Problem Statement}
\label{sec:Overall_Problem_Statement}

The overall control problem considered in this work can now be stated as follows: Design a control input $u(t)$ such that the state $x(t)$ given by
\begin{equation}
\dot{x}(t)
=
A(\sigma_0)x(t)
+
B(\sigma_0)\Lambda R_u(u(t))
\label{eq:problem_uncertain_plant}
\end{equation}
tracks a commanded reference trajectory $x^{\ast}(t)$, while simultaneously satisfying state constraints \eqref{eq:alpha_crit_interval} which ensure flight envelope protection and the input constraints \eqref{eqn:MRAC_SettlementFunction}. This is desired to be accomplished in the presence of uncertainties in $A(\sigma_0)$ and $\Lambda$. 
In what follows, we will propose an adaptive controller to solve this problem.




\section{Adaptive Safety-Critical Control Design}
\label{Safe_AC}

\subsection{Accommodation of Input Constraints}
\label{Sect:Adaptive Control Design}
In this section, we consider the adaptive control of \eqref{eqn:LinearPlantModel_Problem} when subjected to parametric uncertainties and magnitude limits as in \eqref{eqn:MRAC_SettlementFunction}. Our goal is to establish boundedness and control tracking.
The following assumptions are made regarding the unknown parameters in \eqref{eqn:LinearPlantModel_Problem}:
\begin{assumption}
\label{assumption_matching_condition}
Constant matrices $\theta_x^*$ and $\theta_r^*$ exist that solve the following:
\begin{equation}
\label{matching_condition_1}
A_m = A_p + B_p \Lambda \theta_x^*,
\end{equation}
\begin{equation}
\label{matching_condition_2}
B_m = B_p \Lambda \theta_r^*.
\end{equation}
\end{assumption}
\begin{assumption}
The input uncertainty $\Lambda$ is a diagonal positive definite matrix of the form:
\begin{equation}
    \label{Lambda_assumption}
    \Lambda =
    \begin{bmatrix}
        \lambda_1   & 0    & 0\\
        0      & \lambda_2    & 0\\
        0           & 0    & \lambda_3
    \end{bmatrix}.
\end{equation}
\end{assumption}
In order to accommodate magnitude limits on the input and to improve the transient performance, we consider a CCRM of the form \cite{Autenrieb2023}
\begin{equation}
    \label{CCRM}
    \dot{x}_m = A_m x_m(t) + B_m r^*(t) + Le_x(t) + B_p \hat{\Lambda}\Delta u(t),
\end{equation}
where $L$ is a matrix such that $(A_m-L)$ is Hurwitz, $\Delta u(t)= R_u({u}(t))-u(t)$ and represents a disturbance due to saturation, and $\hat{\Lambda}$ is an estimation of the unknown matrix $\Lambda$. 
The idea behind the CCRM is to introduce a feedback from the state error to help improve the transients and from the saturated input signal to accommodate the effects of saturation, represented by  the last two terms in \eqref{CCRM}. When  these terms are zero, the command signal $r^*(t)$ in \eqref{CCRM} is chosen so that $x_m(t)$ reproduces the desired tracking trajectory \(x^*(t)\). This can be achieved by choosing
\begin{equation}
\label{eqn:reference_signal_computing}
    r^*(t)=B_m^\dagger\left(\dot{x}^*(t)-A_m x^*(t)\right),
\end{equation}
where \(B_m^\dagger\) denotes the  pseudoinverse of \(B_m\). 

We propose the following adaptive controller for determining the control input in the plant \eqref{eqn:LinearPlantModel_Problem}:
\begin{equation}
\label{adaptive_feeback_controller}
    u(t)  = \widehat{\theta}_x(t) x_p(t) + \widehat{\theta}_r(t) r^*(t)
\end{equation}
with the time-varying parameters 
adjusted as
\begin{equation}
\label{theta_x_adaption_law}
    \dot{\widehat{\theta}}_x(t) = - \Gamma_x x_p(t) e_x(t)^T P B_p\,, \,\,\,\,\,\, \Gamma_x>0,
\end{equation}
\begin{equation}
\label{theta_r_adaption_law}
    \dot{\widehat{\theta}}_r(t) = - \Gamma_r r^*(t) e_x(t)^T P B_p\,, \,\,\,\,\,\, \Gamma_r>0,
\end{equation}
\begin{equation}
    \label{lambda_r_adaption_law}
    \dot{\hat{\Lambda}}(t) = \Gamma_{\Lambda} \Delta u(t) e_x^T(t) P B_p\,, \,\,\,\,\,\, \Gamma_{\Lambda}>0,
\end{equation}
where $e_x(t)=x_p(t)-x_m(t)$ and $P$ is the solution of the Lyapunov equation $A_m^T P + P A_m = -Q$,  $Q>0$. The adaptive gains $\Gamma_x$, $\Gamma_r$, and $\Gamma_{\Lambda}$  are chosen to be positive definite matrices.

Using \eqref{eqn:LinearPlantModel_Problem}, \eqref{matching_condition_1}, \eqref{matching_condition_2}, \eqref{CCRM}, and \eqref{adaptive_feeback_controller}, we obtain the error dynamics
\cite{Gibson2013Access, Gaudio_2022}:
\begin{equation}
    \label{error_dyanmics_CCRM}
    \begin{split}
    \dot{e}_x(t) = (A_m - L) e_x(t) &+ B_p \Lambda (\Tilde{\theta}_x(t) x_p(t)\\ + \Tilde{\theta}_r(t) r^*(t)) &+ B_p \tilde{\Lambda}(t) \Delta u(t),
    \end{split}
\end{equation}
where $\tilde{\Lambda}(t) = \Lambda - \hat{\Lambda}(t)$ is the corresponding estimation error for $\Lambda$. 
It is easy to show from the structure of \eqref{error_dyanmics_CCRM} that the following $V$ is a Lyapunov function:
\begin{equation}
\begin{split}    
    \label{Lyapunov_function_CCRM}
    V      = &\frac{1}{2} e_x(t)^T P e_x(t) + \frac{1}{2} \operatorname{Tr}  [ \Tilde{\theta}_x(t) \Gamma_1^{-1} \Tilde{\theta}_x^T(t) \Lambda ]\\   & + \frac{1}{2} \operatorname{Tr}  [ \Tilde{\theta}_r(t) \Gamma_2^{-1} \Tilde{\theta}_r^T(t) \Lambda ]  + \frac{1}{2} \operatorname{Tr}  [ \tilde{\Lambda}(t)  \Gamma_{\Lambda}^{-1} \tilde{\Lambda}^T(t) ].
\end{split}
\end{equation}
This follows from \eqref{error_dyanmics_CCRM}, \eqref{theta_x_adaption_law}, \eqref{theta_r_adaption_law}, \eqref{lambda_r_adaption_law} and the fact that $(A_m - L)^T P +  P (A_m - L) = -Q_0$, as $\dot{V} =  - \frac{1}{2}e_x^T Q_0 e_x \leq 0$. This ensures the boundedness of $e_x(t)$ and all parameter estimates $\widehat{\theta}_x(t)$, $\widehat{\theta}_r(t)$, $\widehat{\Lambda}(t)$. 

As the input is subject to saturation, one cannot establish the global boundedness of $x_p(t)$ for arbitrary plants. However, the boundedness of $x_p(t)$ can be established for all initial conditions of $x_p(t)$, $\Tilde{\theta}_x(t)$, $\Tilde{\theta}_r(t)$, $\tilde{\Lambda}(t)$ that lie inside a compact set, similar to \cite{Gaudio_2022,karason1994,Schwager_2005}. Details of this proof are omitted due to space limitations. Such a domain of attraction is a standard requirement for any plant that can be open-loop unstable \cite{karason1994}, an example of which is the GHGV-2. 

\subsubsection{Simulation Results}
\label{subsec:Results_Adaptive_Controllers}
We evaluate the performance and choice of the adaptive controller in ~\eqref{adaptive_feeback_controller}-\eqref{lambda_r_adaption_law} using a nonlinear model of the GHGV-2 which is described in Appendix \ref{FirstAppendix} and implemented using \textsc{MATLAB}/Simulink. The operating point was chosen as $\sigma_0=[12.5, 40\,km]^T$. The initial state was chosen to be $x=[0,0,0,0,0,0]^T$. A $20\%$ uncertainty is assumed to be present in $C_{M,\alpha}$ and $C_{M,q}$, together with a $20\%$ loss in pitch control effectiveness $\lambda_2$, defined in \eqref{Lambda_assumption}. A reference input (black dash-dotted) was chosen, based on \eqref{eqn:reference_signal_computing}, following a step command with the magnitude of $4^\circ$ in the $\alpha$ channel of $x^*(t)$. The baseline controller 
is chosen using a conventional LQR method  \cite{Kalman_1960LQR}, with the following weighting matrices $Q=\mathop{diag}([1,1,1,1,1,1])$ and $R=\mathop{diag}([0.1,0.1,0.1])$. The controller synthesis delivers the gain matrices $K_x$ and $K_r$, which are also used as the initial estimates of the adaptive gains $\theta_x$ and $\theta_r$. The reference model is selected via $A_m=A+BK_x$ and $B_m=BK_r$. 
The control input was assumed to be magnitude-limited as specified in \eqref{eqn:control volume moment} to \eqref{eqn:g(u)}, with $u_{\max,y}$ set to $780\,Nm$, corresponding to the maximum pitching moment. 



\begin{figure}
    \centering
    \includegraphics[width=\columnwidth]{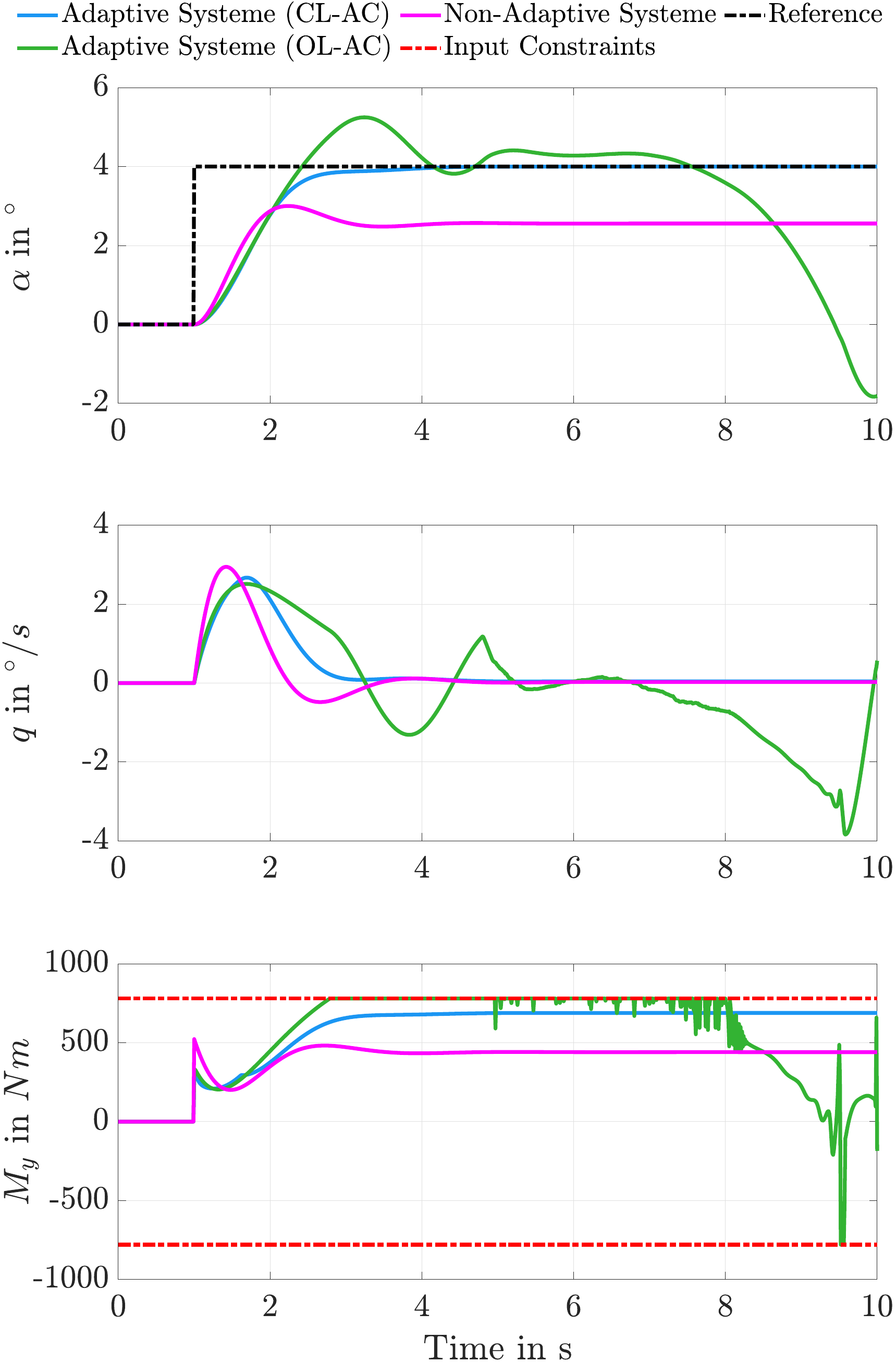}
    \caption{Comparison of the proposed adaptive system with calibrated closed-loop reference model (CL-AC) versus the adaptive system with open-loop reference model (OL-AC) and the non-adaptive baseline controller.}
    \label{fig:AC_saturation_OLvsCL}
\end{figure}

Figure~\ref{fig:AC_saturation_OLvsCL} illustrates the resulting responses of the proposed adaptive controller with the CCRM, denoted as CL-AC. This response is compared with an adaptive system with an open-loop reference model, denoted as OL-AC, which corresponds to the reference model in \eqref{CCRM} with $L$ and $\Delta u$ set to zero, which is commonly used 
in many adaptive control applications. 
As can be seen, the OL-AC fails to follow the desired $\alpha^*$ command. It first overshoots the target and subsequently becomes unstable under the tightened input constraints. 
In contrast, the proposed CL-AC (blue), which accounts for possible saturation through updated estimates of $\widehat\Lambda$ in the calibrated closed-loop reference model, remains stable and converges within the imposed limits. The non-adaptive baseline (magenta) remains stable but fails to reach the desired $\alpha^*(t)$ due to uncompensated model mismatch, resulting in a steady-state error caused by insufficient control authority.


These results clearly show that the proposed adaptive controller in ~\eqref{adaptive_feeback_controller}-\eqref{lambda_r_adaption_law} achieves satisfactory tracking performance and maintains overall stability while accommodating magnitude constraints on the control input. In the next section, we direct our attention to ensuring flight envelope protection.

\subsection{CBF-based Flight Envelope Protection Design}\label{sec:cbf}

In this section, we focus on a control architecture that ensures flight envelope protection in the presence of parametric uncertainties in the model of the GHGV-2. For this purpose, we integrate a CBF-based approach with an adaptive controller (as shown in Fig.~\ref{fig:Closed-LoopAdaptiveControlSystem}). 

\begin{figure*}
    \centering
    \includegraphics[width=0.9\textwidth]{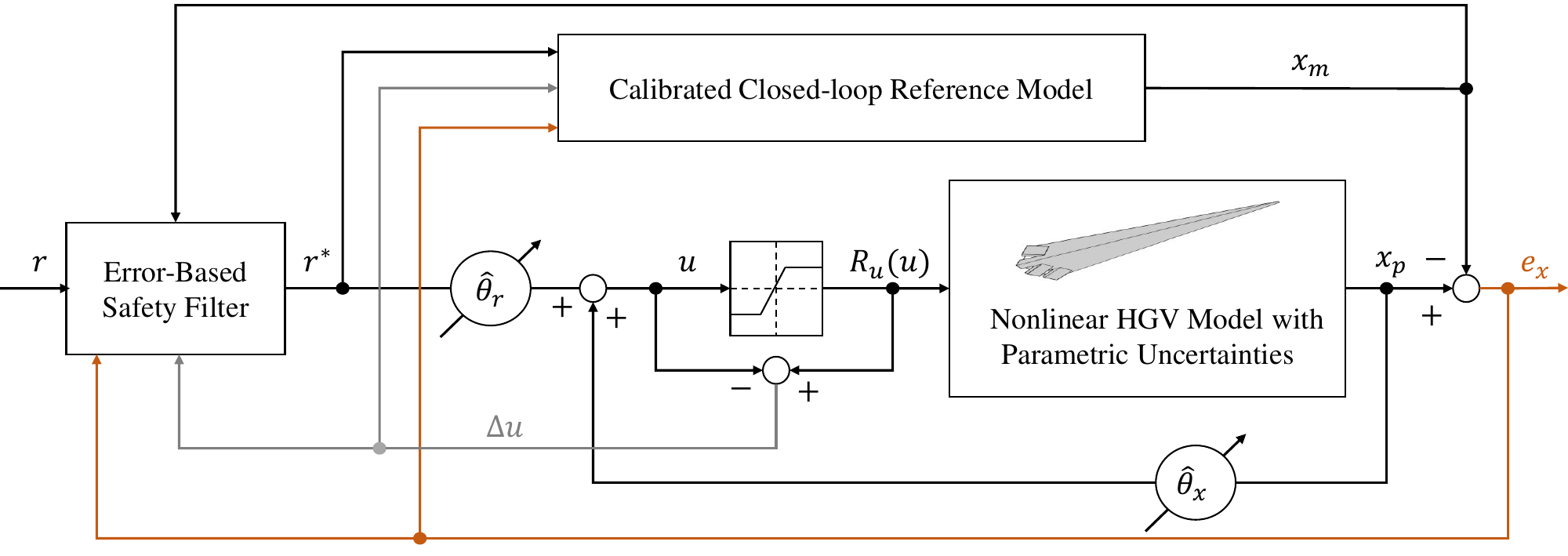}
    \caption{Overview of the proposed adaptive safety-critical control architecture.}
    \label{fig:Closed-LoopAdaptiveControlSystem}
\end{figure*}

We first discuss an architecture that uses a static error buffer (SEB) approach, which keeps the system at all times in a restricted subset of the flight envelope \cite{Autenrieb2023}. We then present a second approach for guaranteed flight envelope protection that uses an Error-Based Safety Filter (EBSF) \cite{fisher_2025} which significantly relaxes the safety condition used in the SEB approach.

\subsubsection{Static Error Buffer (SEB)} \label{subsubsec:SEB}
The fundamental concept of CBFs is to ensure that a set $S \subset \mathbb{R}^n$ which corresponds to a safe behavior, is rendered forward invariant, and to express $S$ using a smooth function $h : \mathbb{R}^n \to \mathbb{R}$ as
\begin{equation}
S = \{x \in \mathbb{R}^n : h(x) \geq 0\}.\label{eq:hx-cbf}
\end{equation}  
Then, forward invariance, which corresponds to the safety objective of $x_p(t) \in S\ \forall t \geq t_0$ becomes equivalent to ensuring that $h(x_p(t)) \geq 0\ \forall t \geq t_0$. In an adaptive control setting, the goal is to  ensure that $S$ is forward invariant in the presence of parametric uncertainties. The idea behind the static safety buffer approach in \cite{Autenrieb2023} is to design a CBF-based safety filter \cite{Ames_2019}, such that the reference model solutions are forward invariant in a constrained set $S_E$ defined later, and to design the closed-loop adaptive system such that its solution converges to that of the reference model. The static buffer corresponds to an addition $\Delta$ in the safety filter so that the reference model state $x_m$ is such that $h(x_m)\geq \Delta$ which ensures that the plant state satisfies \eqref{eq:hx-cbf}. We present the SEB solution below. 

A CBF filter is designed as the solution $r^*(t)$ of the following optimization problem
\begin{equation}
\label{eqn:safety_filter_SEB}
\begin{split}
\argmin_{r^*(t) \in \mathbb{R}^m}
\;\|{r}_{cmd}(t) - r^*(t)\|^2 \quad
\text{s.t.} \quad \eqref{CBF_constraint_reference_model_extended_safety_buffer}
\end{split}
\end{equation}
\begin{equation}
\label{CBF_constraint_reference_model_extended_safety_buffer}
\begin{split}
\frac{\partial h}{\partial x}&\Big|_{x_m(t)} \big[ A_m x_m(t) + B_m r^*(t) + Le_x(t) \\
&+ B_p \hat{\Lambda}(t)\Delta u(t) \big] \geq - \gamma(h(x_m(t)) - \Delta)
\end{split}
\end{equation}
In the above, $\gamma>0$ is a scalar constant. The idea behind the addition of $\Delta$ is that the inequality \eqref{CBF_constraint_reference_model_extended_safety_buffer} ensures that the reference model solution remains in $S_E$ given by
\begin{equation}
\label{safety_subset}
S_E \triangleq \{ x_m(t) \in \mathbb{R}^n \mid h(x_m(t)) \geq E \} \subset S,
\end{equation} 
where 
\begin{equation}
\label{safetyerrorbound}
E = \sup_{t \geq t_0} |h(x_p(t))-h(x_m(t))|
\end{equation} 
which can be shown to be finite through appropriate adaptive control designs. What remains is to  choose $\Delta$ such that $x_m\in S_E$ guarantees that $x_p(t)\in S$ for all $t\geq t_0$, by leveraging the Lipschitz property of $h(\cdot)$. 

This was established in \cite{Autenrieb2023} and is summarized in the theorem below:
\begin{theorem}
\label{Theorem_Safety_Buffer}
Consider the overall closed-loop adaptive system defined by the plant in \eqref{eqn:LinearPlantModel_Problem}, the reference model in \eqref{CCRM}, the control input in \eqref{adaptive_feeback_controller}, and the adaptation laws in \eqref{theta_x_adaption_law}-\eqref{theta_r_adaption_law} under the initial conditions $x_m(t_0) \in S_E$ and $x_p(t_0) \in S$. Let $r^*(t)$ satisfy the safety filter in \eqref{eqn:safety_filter_SEB} for every $t \geq t_0$ with any $\Delta \geq E$ for $E$ defined in \eqref{safetyerrorbound}. Then, we have $x_m(t) \in S_E$ and $x_p(t) \in S\ \forall t \geq t_0$.
\end{theorem}
\noindent
\begin{proof}
In accordance with Nagumo's theorem, the forward invariance of \( x_m(t)\)  within \(S_E \) can be ensured via $\dot{h}(x_m(t)) \geq 0 \quad \forall x_m(t) \in \partial S_E$, where $\partial S_E$ is the boundary of the safe subset $S_E$. Based on \eqref{CBF_constraint_reference_model_extended_safety_buffer}, we consider the inequality $\dot{h}(x_m(t)) \geq -\gamma(h(x_m(t)) - \Delta)$ to render the system safe. We choose $\Delta$ to be lower bounded by $ \Delta \geq E$. This delivers $\dot{h}(x_m(t)) \geq 0$ at $\partial S_E$, and hence $x_m(t) \in S_E$ for all $t \geq t_0$. By the definition of $E$, it follows that $h(x_p(t)) \geq h(x_m(t)) - E \geq 0$, and therefore $x_p(t) \in S$ for all $t \geq t_0$.

\end{proof}

While the introduction of a static safety buffer \(\Delta\) provides a straightforward mechanism to enforce constraint satisfaction by uniformly shrinking the admissible set of reference commands, it often results in overly conservative system behavior.
This is because the reference state $x_m(t)$ is constrained to remain in $S_E$ for all $t\geq t_0$, which compromises the tracking performance of the closed-loop system during adaptation. Thus, the next approach, EBSF, reduces conservatism by making the following observation: at times $t$ when $e_x(t)$ is small, the plant is closely tracking the reference model, and thus safety of the reference dynamics in \eqref{CCRM} is equivalent to safety of the closed-loop plant. Thus, at such times when $e_x(t)$ is small, the safety buffer can be removed and $x_m(t)$ can approach the boundary of $S$ more closely without jeopardizing the safety of the closed-loop plant.

\subsubsection{Error-Based Safety Filter (EBSF)}\label{subsubsec:EBSF}


For ease of exposition, we assume that $B_p$ is known, and set $B_m=B_p$. We begin with the following assumptions:


\begin{assumption} \label{asn:theta_x_in_set}
$\theta_x^* \in \Theta_x$ for a known compact, convex set $\Theta_x \subset \mathbb{R}^{m \times n}$.
\end{assumption}


\begin{assumption} \label{asn:lambda_in_set}
$\Lambda = \mathrm{diag}(\lambda_*)$ for a vector $\lambda_* \in \mathbb{R}^m$ which satisfies $\lambda_* \in L := \{\lambda \in \mathbb{R}^m : \underline{\lambda}_i \leq \lambda_i \leq \overline{\lambda}_i\ \forall i\}$ for known $\underline{\lambda}_i, \overline{\lambda}_i > 0$. We denote this as $\Lambda \in \mathcal{L}$.
\end{assumption}

\noindent
We choose
the feedforward gain estimate $\hat{\theta}_r(t)$ in \eqref{adaptive_feeback_controller} as
\begin{equation} \label{eqn:vartheta_r}
    \hat{\theta}_r(t) = \mathrm{diag}(\hat{\vartheta}_r(t)),
\end{equation}
where $\hat{\vartheta}_r(t) \in \mathbb{R}^m$ estimates $1/\lambda_*$ element-wise and has the adaptive law
\begin{equation}
\label{vartheta_r_adaption_law}
\begin{aligned}
    \overline{\vartheta}_r &= - \Gamma_r\mathrm{diag}(r)B_p^\top Pe_x, \\
    \dot{\widehat{\vartheta}}_{r,i} &= \begin{cases} 0, & \begin{matrix*}[l] (\widehat{\vartheta}_{r,i} \geq \frac{1}{\underline{\lambda}_i} \land \overline{\vartheta}_{r,i} \geq 0) \\
    \lor (\widehat{\vartheta}_{r,i} \leq \frac{1}{\overline{\lambda}_i} \land \overline{\vartheta}_{r,i} \leq 0) \end{matrix*} \\ \overline{\vartheta}_{r,i}, & \mathrm{otherwise} \end{cases}\ \forall i.
\end{aligned}
\end{equation}
It is easy to show that the adaptive control analysis in Section \ref{Sect:Adaptive Control Design} still holds by appropriately adjusting the Lyapunov function in \eqref{Lyapunov_function_CCRM}.




In order to apply EBSF to the  flight envelope protection problem discussed in Section \ref{Problem Formulation}, we 
start with \eqref{eq:S_alpha_intersection} and define the following safe regions in state-space:
\begin{equation}
    \label{eqn:S_1}
      S_1(\sigma_0) = \{ x  \in \mathbb{R}^n : h_1(x; \sigma_0) \geq 0 \},
\end{equation}
    \begin{equation}
    \label{eqn:S_2}
    S_2(\sigma_0) = \{ x  \in \mathbb{R}^n : h_2(x; \sigma_0) \geq 0 \},
    \end{equation} where CBF candidates $h_1$ and $h_2$ are defined as
\begin{equation}
    \label{eqn:h_1}
     h_1(x; \sigma_0) = \alpha_{\mathrm{crit,\max}}(\sigma_0)  - \alpha,
\end{equation}
\begin{equation}
\label{eqn:h_2}
h_2(x; \sigma_0) = \alpha  - \alpha_{\mathrm{crit,\min}}(\sigma_0),
\end{equation}
\noindent
 noting that $\alpha$ is an element of the state vector $x \in \mathbb{R}^n$.
The safety objective is to enforce forward invariance of $S(\sigma_0) := S_1(\sigma_0) \cap S_2(\sigma_0)$. In what follows, we will drop the dependence on $\sigma_0$ for ease of exposition.

As $h_1$ and $h_2$ in \eqref{eqn:h_1}-\eqref{eqn:h_2} have relative degree 2 with respect to the dynamics in \eqref{eq:problem_uncertain_plant}, we construct high-order CBFs $h_{12}(x)$ and $h_{22}(x)$ from $h_1$ and $h_2$, respectively, as in \cite{fisher_2025}. Then, we choose $r^*$ according to \eqref{eqn:governor_input_uncertainty} for any constants $\gamma_1, \gamma_2 > 0$, $\delta_1 \in (0, \sup_{{x} \in S} \gamma_1h_{12}({x}))$, and $\delta_2 \in (0, \sup_{{x} \in S} \gamma_2h_{22}({x}))$.
\begin{figure*}[h!]
\begin{equation}
\label{eqn:governor_input_uncertainty}
\begin{split}
&\argmin_{{r}^* \in \mathbb{R}^m}
\;\|{r}_{cmd} - {r}^*\|^2 \\
\text{s.t.}\quad
&
\min_{\theta_x \in \Theta_x,\;\lambda \in L}
\frac{\partial h_{12}}{\partial {x}}\Big|_{{x}_p}
\Big[
(1 - \beta_1({x}_p, {e}_x))\,{z}_m({x}_m, r, e_x, \hat{\Lambda}, \Delta u) \\
&\hspace{3.2em}
+ \beta_1({x}_p, {e}_x)\,
{z}_p(x_p, r, \hat{\theta}_x, \hat{\theta}_r, \theta_x, \lambda, \Delta u)
\Big]
\geq -\gamma_1 h_1({x}_p) + \delta_1,\\
&
\min_{\theta_x \in \Theta_x,\;\lambda \in L}
\frac{\partial h_{22}}{\partial {x}}\Big|_{{x}_p}
\Big[
(1 - \beta_2({x}_p, {e}_x))\,{z}_m({x}_m, r, e_x, \hat{\Lambda}, \Delta u) \\
&\hspace{3.2em}
+ \beta_2({x}_p, {e}_x)\,
{z}_p(x_p, r, \hat{\theta}_x, \hat{\theta}_r, \theta_x, \lambda, \Delta u)
\Big]
\geq -\gamma_2 h_2({x}_p) + \delta_2
\end{split}
\end{equation}
\end{figure*}

The vectors ${z}_m$ and ${z}_p$ are given by
\begin{equation}
\begin{split}
\label{eqn:z_m}
&{z}_m({x}, {r}, e_x, \hat{\Lambda}, \Delta u)
= {} \\
&A_m x + B_mr + L e_x + B_p \hat{\Lambda}\Delta u
\end{split}
\end{equation}
and
\begin{equation}
\label{eqn:z_p}
\begin{split}
&{z}_p({x}, {r}, \hat{\theta}_x, \hat{\theta}_r, \theta_x, \lambda, \Delta{u})
= {} \\
&A_mx + B_p\mathrm{diag}(\lambda)((\hat{\theta}_x - \theta_x) x + \hat{\theta}_r r + \Delta{u}).
\end{split}
\end{equation}
Furthermore,
$\beta_1, \beta_2 : \mathbb{R}^n \times \mathbb{R}^n \to [0,1]$
are state-dependent interpolation parameters that are designed to smoothly
transition between enforcing safety with respect to the nominal reference
dynamics and enforcing safety for the plant under worst-case parametric
uncertainty, and are given by
\begin{equation*}
\label{eqn:ebcg_interpolation}
\beta_j({x}, {e}_x)
=
\mathrm{sat}\!\left(
\frac{
\max\!\left\{
\kappa_j\|e_x\|,
\frac{2\delta_j}{3\gamma_j}
\right\}
- h_{j2}({x})
}{
\max\!\left\{
\kappa\|e_x\|,
\frac{2\delta_j}{3\gamma_j}
\right\}
- \frac{\delta_j}{3\gamma_j}
}
\right),
\end{equation*}
for each $j \in \{1, 2\}$. The function $\mathrm{sat}(\cdot)$ is given by
\begin{equation*} \label{eqn:sat}
    \mathrm{sat}(z) = \begin{cases} 0, & z \leq 0 \\ z, & 0 < z < 1 \\ 1, & z \geq 1 \end{cases},
\end{equation*}
and $\kappa_1, \kappa_2 > 0$ are any constants such that $\sup_{t \geq 0} \kappa_j\|e_x\| < \sup_{{x} \in S} h_{j2}({x})\ \forall j \in \{1, 2\}$. As previously discussed, we know from the adaptive control analysis in Section \ref{Sect:Adaptive Control Design} that $\|e_x(t)\|$, and can derive an upper bound on $\sup_{t \geq 0} \|e_x\|$ from \eqref{Lyapunov_function_CCRM} and bounds on the uncertain parameters.

\begin{remark}
We note that, for the integrated safety filter in \eqref{eqn:governor_input_uncertainty}, the upper and lower bounds on $\alpha$ are enforced via two distinct inequality constraints. In general, when multiple CBF-based constraints are put into a single QP as in \eqref{eqn:governor_input_uncertainty}, it is difficult to guarantee feasibility of the QP for all time. However, the structure of the considered FEP problem allows us to avoid infeasibility in practice. Specifically, at any given time, the system is approaching either the upper or the lower bound on $\alpha$, but never both at once. Thus, by properly tuning the hyperparameters $\gamma_1, \gamma_2, \delta_1, \delta_2, \kappa_1, \kappa_2 > 0$, we can ensure that the QP is always feasible.
\end{remark}

\begin{theorem}
\label{Theorem_EBSF}
Consider the overall closed-loop adaptive system defined by the plant in \eqref{eqn:LinearPlantModel_Problem}, the reference model in \eqref{CCRM}, the control input in \eqref{adaptive_feeback_controller}, and the adaptation laws in \eqref{theta_x_adaption_law} and \eqref{eqn:vartheta_r}-\eqref{vartheta_r_adaption_law} under the initial conditions $x_p(t_0) \in S$ and $\hat{\vartheta}_r(t_0)^{-1} \in L$, where the inverse is taken element-wise. Let $r^*(t)$ satisfy the safety filter in \eqref{eqn:governor_input_uncertainty} for every $t \geq t_0$. Then, we have $x_p(t) \in S\ \forall t \geq t_0$.
\end{theorem}
\noindent
We refer the reader to \cite{fisher_2025} for the proof of Theorem~\ref{Theorem_EBSF}.

\section{Numerical Studies}
\label{Numerical_Assessment}
\subsection{Single Operating-Point Analysis}
The following results assess and compare the satisfaction of the state constraints for different control architectures for the GHGV-2 model described in Appendix~\ref{FirstAppendix} using \textsc{MATLAB}/Simulink. In all simulations, the plant is controlled by the same baseline and adaptive controllers described in Section~\ref{Sect:Adaptive Control Design} and Section~\ref{sec:cbf}. The operating point was chosen as $\sigma_0=[8, 30\,km]$, and the initial state as $x=[0,0,0,0,0,0]^T$. To evaluate the effectiveness of the proposed EBSF method, different safety filter configurations are compared. The blue curve represents the baseline case without any safety filtering, where the reference command is applied directly to the controller without modification. The green curve corresponds to safety enforcement through reference clipping. In this approach, reference commands that would drive the system outside the admissible state set are simply cut off, or saturated, at the admissible boundary before being passed to the controller. The orange curve shows the response obtained with a classical CBF-based safety filter with no adaptation \cite{Ames_2017}. The purple curve represents the CBF formulation augmented with a static error buffer described in Section \ref{subsubsec:SEB}, denoted as CBF+SEB, where the admissible set is conservatively reduced by a fixed safety margin. Finally, the black curve corresponds to the proposed CBF formulation with EBSF described in Section \ref{subsubsec:EBSF}, denoted as EBSF, where the safety margin is adjusted based on the estimated tracking error between the reference model and the actual system response. A command signal $\alpha_{cmd}(t)$ (black dash-dotted) is chosen as a step with a magnitude of $5^\circ$, while the safety limit is artificially set to $\alpha_{\mathrm{crit,\max}} = 4^\circ$ (red dash-dotted). The ideal case in which there are no uncertainties in the aerodynamics and control effectiveness is presented in  Fig.~\ref{fig:nominal_FEP}, which shows the time histories of the measured angle of attack $\alpha(t)$, the filtered reference $\alpha^*(t)$ (if a safety filter is applied), the pitch rate $q(t)$, and the commanded pitch moment $M_y(t)$. These plots illustrate the feasibility of the nonlinear model to meet state constraints when there are no uncertainties.
\begin{figure}[h!]
    \centering
    \includegraphics[width=\columnwidth]{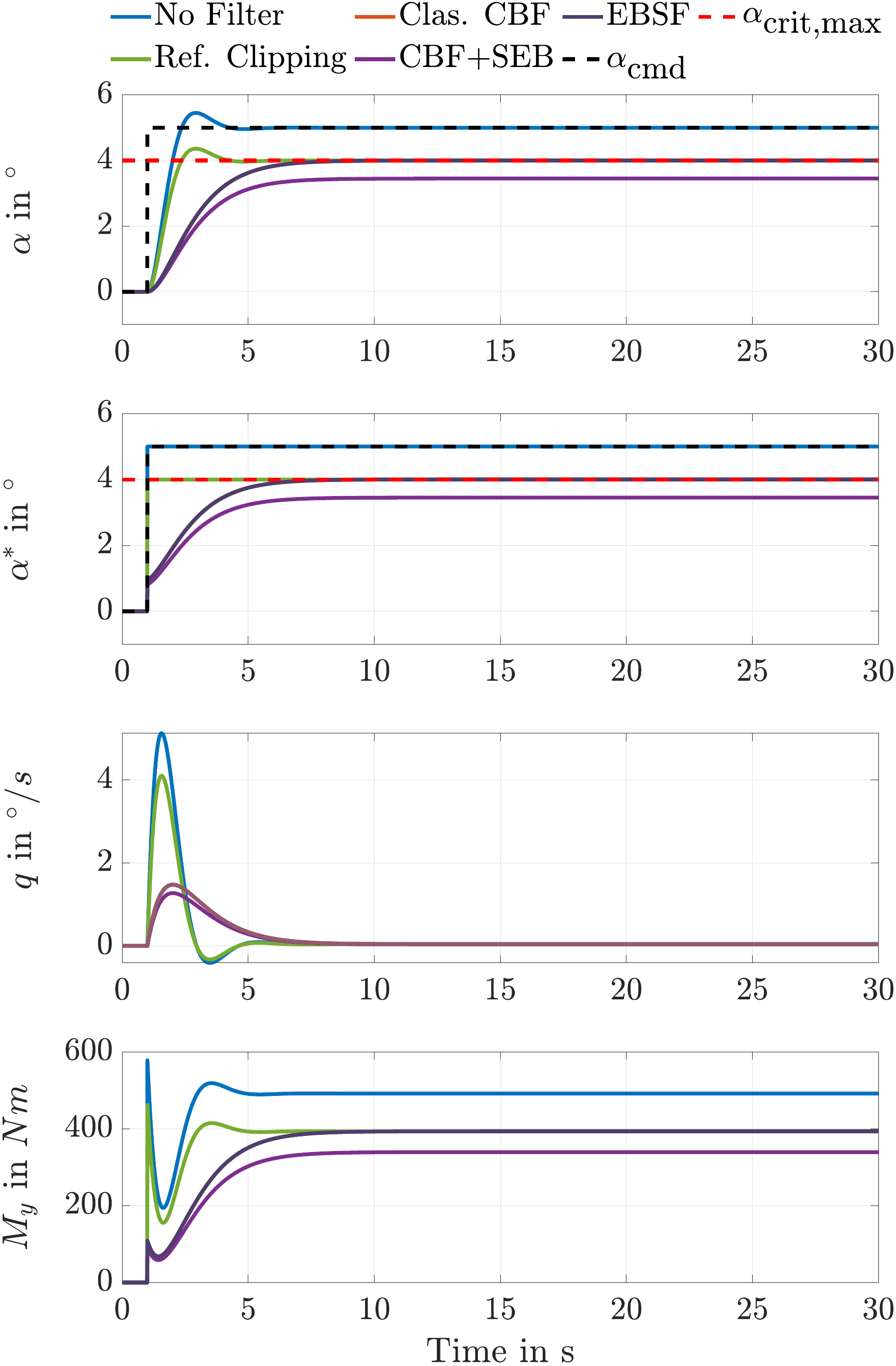}
    \caption{Comparison of the proposed adaptive closed-loop system with different safety-filter methods for a case with no model mismatches.}
    \label{fig:nominal_FEP}
\end{figure}
Several observations can be made: Without a safety filter, the angle of attack violates the state constraint of $4^\circ$ and converges to $5^\circ$. Reference clipping enforces the filtered reference $\alpha_{cmd}(t) \le \alpha_{crit,\max}$ but does not guarantee satisfaction of the state constraint; closed-loop transients cause the trajectory of $\alpha(t)$ to exceed the limit within the time window from 0 to 5 seconds. Both the classical CBF and the EBSF 
ensure that $\alpha(t)$ remains within the admissible set, reaches the boundary, and exhibits nearly identical transients. As expected, the CBF+SEB approach is conservative: the constant buffer, designed for the worst-case model mismatch scenario, enforces convergence strictly within the safe set, resulting in a greater offset from the boundary in both $\alpha(t)$ and $\alpha_{cmd} (t)$. The rate $q(t)$ and moment $M_y(t)$ traces support these mechanisms, meaning that the safety filter incorporates the higher-order dynamics and reduces $q(t)$ early in the transient to anticipate the second-order behavior and prevent boundary violation. $M_y(t)$ is also not close to the magnitude constraints and saturation effects are not in conflict with enforcing state constraints.

We next study the effect of parametric uncertainty. As in Section III, we introduce a 20\% 
mismatch in $C_{M,\alpha}$ and $C_{M,q}$ 
and reduce the pitch control effectiveness by $10\%$. All other simulation settings remain unchanged. The corresponding results are shown in Fig.~\ref{fig:uncertain_FEP}. Similar comparisons are made as in Fig.~\ref{fig:nominal_FEP}.

\begin{figure}[h!]
    \centering
    \includegraphics[width=\columnwidth]{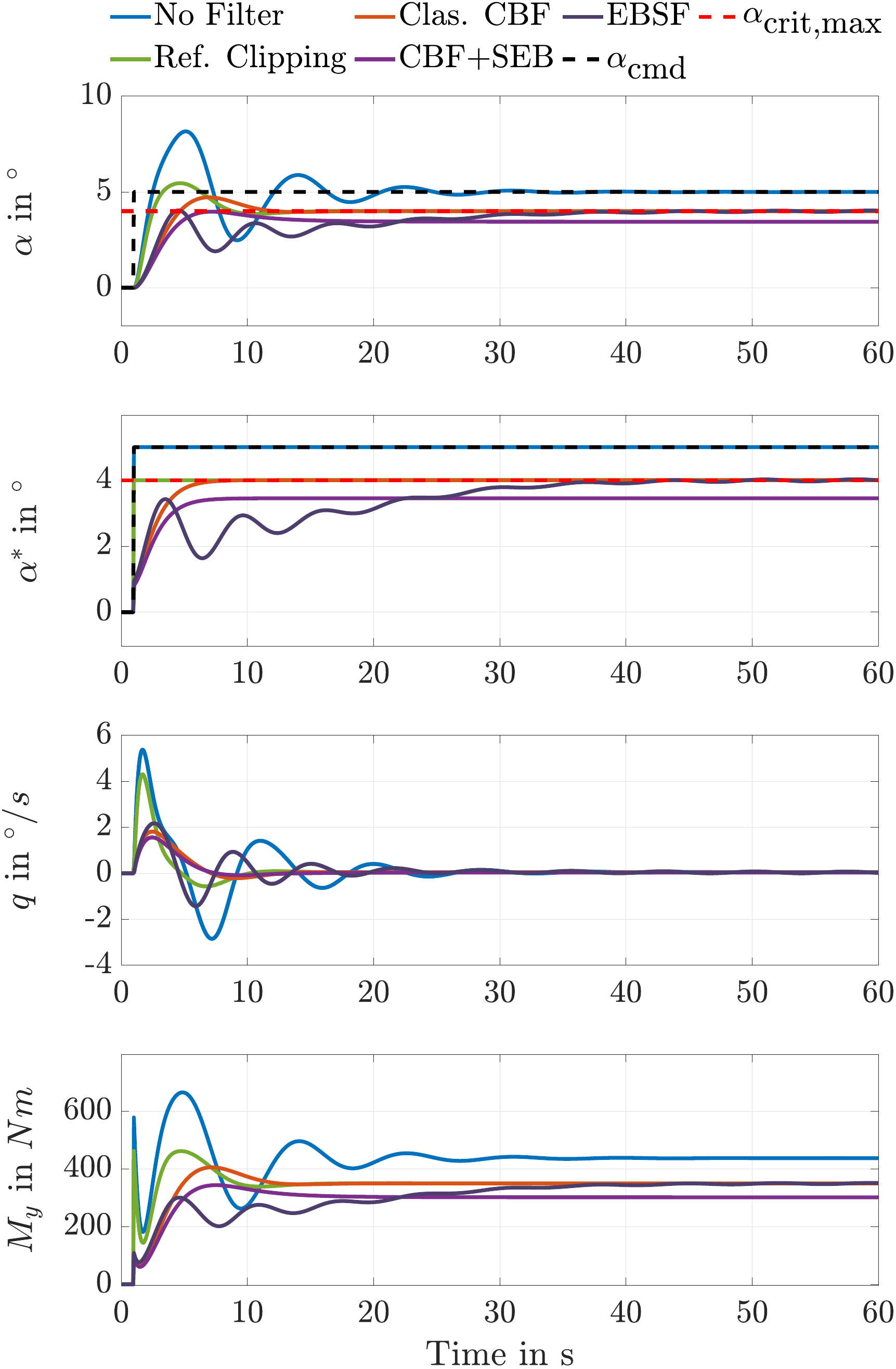}
    \caption{Comparison of the proposed adaptive closed-loop system with different safety-filter methods for a case with model mismatches.}
    \label{fig:uncertain_FEP}
\end{figure}

In this case, as in Fig.~\ref{fig:nominal_FEP}, the system without filtering  violates the constraint significantly and converges to $5^\circ$. Reference clipping also fails to ensure safety, with an overshoot of $\alpha(t)$ beyond $4^\circ$. The classical CBF reduces the overshoot but still exhibits a limit violation, because the safety condition is evaluated on the nominal reference dynamics and does not account for parametric uncertainties. CBF+SEB remains safe due to the steady-state offset from the boundary, which also can be seen in the reduced control effort in $M_y(t)$ and in the corresponding rate $q(t)$. Furthermore, it can be  seen that CBF+SEB has to add additional conservatism in order to address the overshoot dynamics, leading to a steady-state deviation from the safety boundary. EBSF, in contrast, remains safe while asymptotically approaching the boundary after a cautious transient. As the parameter estimates improve and the tracking error $e_x(t)$ decreases, the error-based margin gradually shrinks, allowing the system to recover the requisite tracking performance.

The results show a clear hierarchy among the evaluated methods. Reference clipping modifies only the command and therefore cannot guarantee state safety. The classical CBF ensures safety only when the reference model is sufficiently accurate. CBF+SEB remains safe under uncertainty, but its fixed buffer makes the response persistently conservative. In contrast, EBSF provides robustness while reducing conservatism over time: it initially tightens the constraint under uncertainty and then relaxes it as the tracking error \(e_x(t)\) decreases. In the nominal case, it behaves like the classical CBF. Overall, the results demonstrate the effectiveness of the proposed FEP architecture and support the use of an error-based safety filter for robust yet less conservative operation.

\subsection{Constrained Flight Corridor}

The capabilities of the proposed architecture are next evaluated for a case in which an aerodynamic corridor constraint, as introduced in Section~\ref{sec:Aerodynamic_Corridor_Constraints}, is considered.  However, the purpose of this study is not to reproduce a fully mission-accurate hypersonic trajectory, but rather to construct a controlled test case which illustrates the ability of the proposed framework to enforce altitude-related safety constraints in the presence of uncertainties.
Specifically, the hypersonic system is required to satisfy the time-varying constraint
\begin{equation}
H_{\min}(t) \leq H(t) \leq H_{\max}(t),\label{eq:FEP}
\end{equation}
where the corridor bounds are defined relative to the commanded descent trajectory using a fixed width of $\pm 300\,\mathrm{m}$.  

In this simulation study, the GHGV-2 model is initialized at an operating point $\sigma_0=[6, 20\,km]^T$. Thereafter, the vehicle is trimmed along a desired altitude trajectory. The scenario considered corresponds to a descent maneuver from $20\,\mathrm{km}$ to $17\,\mathrm{km}$ over a time horizon of $60\,\mathrm{s}$, followed by steady flight. As mentioned earlier, the admissible flight corridor is defined such that it follows the nominal trajectory path with an additional altitude safety band of $\pm 300\,\mathrm{m}$. 
To assess the impact of the proposed adaptive controller, an upset event is introduced at $t=40\,\mathrm{s}$ by switching from the nominal plant model to a damaged system with altered aerodynamic characteristics, mass properties, and reduced control effectiveness. This induces a model mismatch that directly affects the influence of lift on the dynamics and thus the altitude evolution. In the studies reported below, the adaptive controller remains active in all cases, while the proposed EBSF-based safety filter is selectively enabled to isolate its effect on constraint enforcement.

\begin{figure}
    \centering
    \includegraphics[width=\columnwidth]{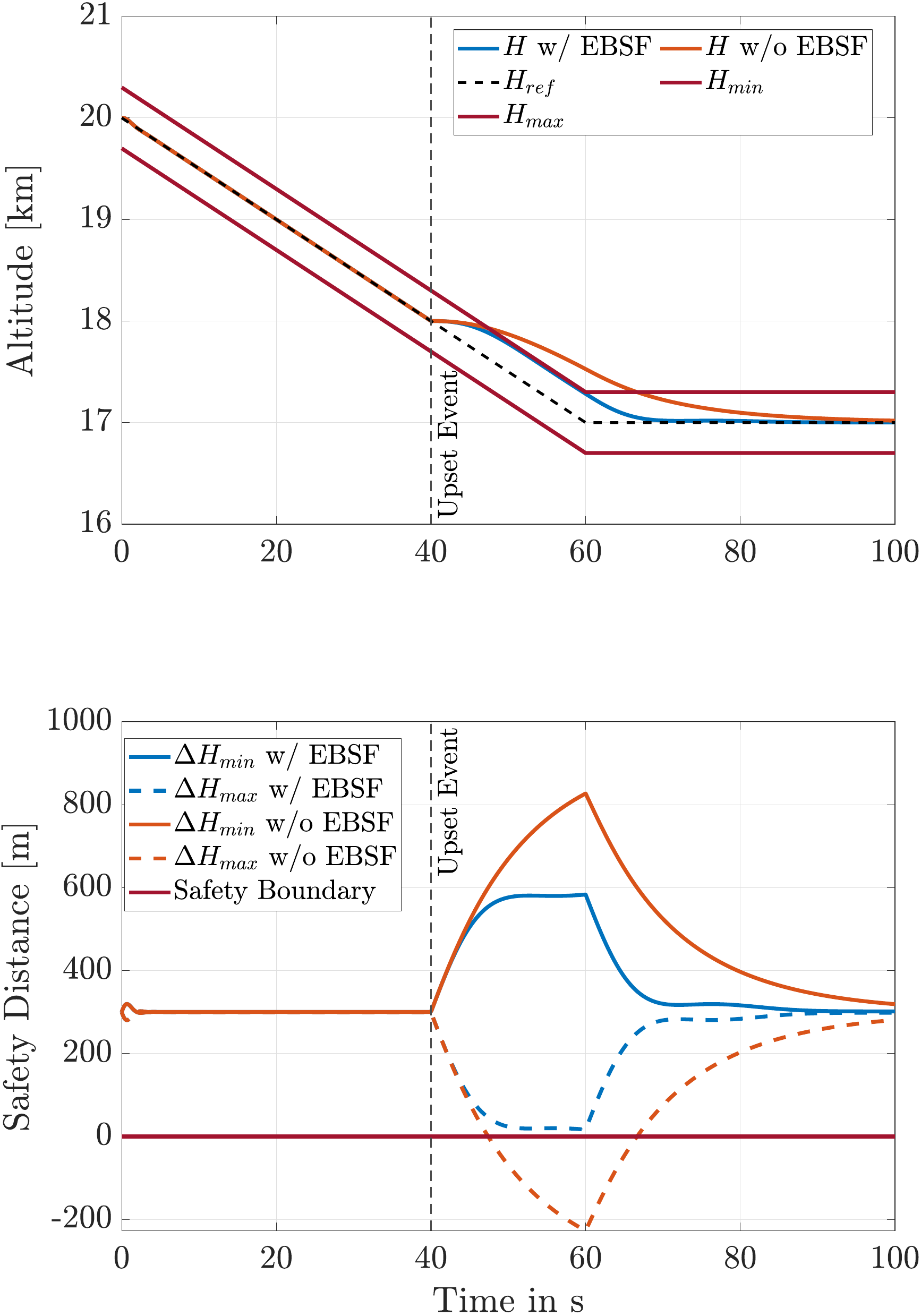}
    \caption{Results of flight corridor case with and without EBSF.}
    \label{fig:corridor_with_and_without_filter}
\end{figure}

Fig.~\ref{fig:corridor_with_and_without_filter} shows the resulting closed-loop responses without and with the EBSF for the integrated adaptive closed-loop system. The figures present the altitude trajectory together with the corridor bounds and  the corresponding safety distances to the lower and upper boundaries. Prior to the upset event, the system tracks  the commanded descent accurately, and the altitude remains well within the admissible corridor. This indicates that, under nominal conditions, the adaptive closed-loop system is sufficient to  achieve the desired trajectory tracking performance and, in addition,  render the system safe during tracking.

Once the damaged dynamics are activated at $t=40\,\mathrm{s}$, a pronounced deviation from the nominal trajectory is observed. In particular, the altitude of the  system without EBSF increases more rapidly than commanded and exceeds the upper corridor boundary $H_{\max}(t)$. This violation is clearly reflected in the safety distance $\Delta H_{\max}(t)=H_{\max}(t)-H(t)$, which becomes negative over a finite time interval. This means the closed-loop system fails to maintain forward invariance of the admissible set when EBSF is not included. The underlying mechanism can be explained by the altitude dynamics derived in \eqref{eq:H_ddot_alpha}. Following the upset event, the damaged dynamics lead to a situation in which the generated lift is no longer compatible with the desired flight path evolution. In particular, the altered dynamics produce an excessive climb tendency relative to the commanded descent trajectory, causing the vehicle to deviate upward from the prescribed corridor. As a result, the altitude increases beyond the admissible region before the controller is able to fully compensate for the mismatch. Although the adaptive controller begins adjusting its parameters in response to the emerging tracking error, this adaptation occurs on a slower time scale and is therefore insufficient to avoid the state constraint violation.

In contrast to the unfiltered case, the system with EBSF remains within the prescribed corridor for the entire simulation horizon, and both safety distances remain non-negative throughout the maneuver. This demonstrates that the proposed safety filter successfully preserves forward invariance of the admissible altitude set. In particular, it should be noted that the safety filter  modifies the closed-loop behavior following the upset event through proper integration and interaction with the adaptive controller. Consequently, the closed-loop system reacts significantly earlier than in the unfiltered case, preventing the large altitude excursion observed previously. 


The comparison highlights the limitation of adaptive tracking alone in safety-critical scenarios. It should also be noted that for both cases, with and without EBSF, the actuator magnitude limits were not critical. Hence, the observed differences between the compared responses are not caused by saturation effects, but by the safety-filter mechanisms and their ability to account for closed-loop transients and model mismatch. Although the adaptive controller recovers the altitude after the upset event, it cannot guarantee corridor satisfaction during the transient phase. In the unfiltered case, the altitude temporarily leaves the admissible corridor and the safety distance becomes negative. With the proposed safety filter, the altitude remains within the prescribed bounds, and the safety distance stays non-negative despite model mismatch and disturbance. Thus, the EBSF provides a reliable safety layer on top of the adaptive controller, ensuring corridor compliance while intervening only when a constraint violation is predicted. Most importantly, the simultaneous accommodation of adaptation and flight envelope protection in the form of the state constraint \eqref{eq:FEP} is fully realized through the incorporation of EBSF in the control architecture, which was illustrated in this section through the simulation studies of the full nonlinear model of the GHGV-2.
\section{Conclusion}
\label{Conclusions}
This paper presented an adaptive safety-critical control framework for flight envelope protection of overactuated hypersonic glide vehicles under model uncertainty. Magnitude-limited control inputs are accommodated within the adaptive closed-loop system through a calibrated closed-loop reference model. On top of this adaptive control architecture, flight envelope state constraints are enforced by an error-based control barrier function safety filter that adjusts the admissible safe set using the measured plant–reference mismatch. Simulations with the DLR GHGV-2 demonstrate stable tracking under limited control authority while maintaining flight envelope constraints during transient adaptation and model mismatch.

\section*{Acknowledgment}
\label{Acknowledgment}
The first author would like to acknowledge Dr.\,Patrick Gruhn from the DLR Institute of Aerodynamics and Flow Technology for helpful discussions and for supporting the vehicle's high-fidelity modeling.


\bibliographystyle{IEEEtran_bst}
\bibliography{references/ref} 


\appendices 
    \section{}
    \subsection{Nonlinear Flight Dynamics Model}
    \label{FirstAppendix}
    \eqref{eqn:Forces_EOM} and \eqref{eqn:Moments_EOM} present the generalized equations of motion for translation and rotation of a flight system in body coordinates:
\begin{equation}
    \label{eqn:Forces_EOM}
    {R} = \begin{bmatrix}F_x \\ F_y \\ F_z \end{bmatrix} 
    = m \begin{bmatrix} \dot{u} \\ \dot{v} \\ \dot{w} \end{bmatrix} 
    + m \begin{bmatrix} p \\ q \\ r \end{bmatrix} 
    \times \begin{bmatrix} u \\ v \\ w \end{bmatrix},
\end{equation}

\begin{equation}
    \label{eqn:Moments_EOM}
    {Q} = \begin{bmatrix} M_x \\ M_y \\ M_z \end{bmatrix} 
    = {I} \begin{bmatrix} \dot{p} \\ \dot{q} \\ \dot{r} \end{bmatrix} 
    + \begin{bmatrix} p \\ q \\ r \end{bmatrix} 
    \times {I} \begin{bmatrix} p \\ q \\ r \end{bmatrix}.
\end{equation}
Here, \( m \) is the mass, \( {I} \) denotes the moment of inertia matrix, \( (p, q, r)^T \) are the body angular rates, and \( (u, v, w)^T \) are the translational velocity components expressed in the body frame. The angular rates \( (p, q, r)^T \) are typically obtained by integrating \eqref{eqn:Moments_EOM}. 

An important kinematic relationship connects the time derivatives of the flight path bank angle \( \mu \), the angle of attack \( \alpha \), and the sideslip angle \( \beta \) to the body angular rates and the derivatives of the flight path angle \( \gamma \) and track angle \( \chi \):

\begin{align}
\label{eqn:Advanced Rotational Kinematics 1}
\begin{bmatrix} \dot{\mu} \\ \dot{\alpha} \\ \dot{\beta} \end{bmatrix}
&= {T}_1(\alpha, \beta) \begin{bmatrix} p \\ q \\ r \end{bmatrix}
+ {T}_2(\mu, \gamma, \chi) \begin{bmatrix} \dot{\gamma} \\ \dot{\chi} \end{bmatrix} \\
&= 
\begin{bmatrix}
    \displaystyle \frac{\cos\alpha}{\cos\beta} & 0 & \displaystyle \frac{\sin\alpha}{\cos\beta} \\
    \displaystyle -\cos\alpha \tan\beta & 1 & \displaystyle -\sin\alpha \tan\beta \\
    \displaystyle \sin\alpha & 0 & -\cos\alpha
\end{bmatrix}
\begin{bmatrix} p \\ q \\ r \end{bmatrix}\\
&+ 
\begin{bmatrix}
    \displaystyle \cos\mu \tan\beta & \sin\gamma + \sin\mu \tan\beta \cos\gamma \\
    \displaystyle -\frac{\cos\mu}{\cos\beta} & \displaystyle -\frac{\sin\mu \cos\gamma}{\cos\beta} \\
    \displaystyle -\sin\mu & \cos\mu \cos\gamma
\end{bmatrix}
\begin{bmatrix} \dot{\gamma} \\ \dot{\chi} \end{bmatrix}. \nonumber
\end{align}

The relationship in \eqref{eqn:Advanced Rotational Kinematics 1} relates the attitude angle derivatives \( (\dot{\mu}, \dot{\alpha}, \dot{\beta})^T \) to the body angular rates and to the derivatives of the inertial flight path angle \( \dot{\gamma} \) and track angle \( \dot{\chi} \). The differential equations for \( \dot{\gamma} \) and \( \dot{\chi} \) are given by
\begin{equation}
    \dot{\gamma} = \frac{1}{mV} \underbrace{\left[ L \cos(\mu) - W \cos(\gamma) - Y \sin(\mu) \cos(\beta) \right]}_{f_{\dot\gamma}},
\end{equation}
\begin{equation}
    \dot{\chi} = \frac{1}{mV \cos(\gamma)} \underbrace{\left[ L \sin(\mu) + Y \cos(\mu) \cos(\beta) \right]}_{f_{\dot\chi}},
\end{equation}
where \( L \) is the aerodynamic lift force, \( Y \) is the aerodynamic side force, \( W = mg \) is the gravitational force, and \( V \) is the airspeed.  We consider perturbations around an operating point characterized by
\begin{equation}
\label{eq:sigma_def}
{\sigma} \;=\; \begin{bmatrix}\operatorname{Ma}_0 & H_0\end{bmatrix}^{\!\top},
\end{equation}
i.e., \({\sigma}_0\) includes only Mach number and altitude, while \(\alpha\) and \(\beta\) remain explicit states. Let \(V_0 = \operatorname{Ma}_0 a_0\) denote the slowly varying airspeed at altitude \(H_0\). Linearizing \(f_{\dot\gamma}\) and \(f_{\dot\chi}\) about the equilibrium manifold at \(({x}_0,{u}_0,{\sigma}_0)\) yields
\begin{equation*}
\label{Gleich: Perturbation_Gamma}
\delta \dot{\gamma} = w_\gamma \left[ \frac{\partial f_{\dot\gamma}}{\partial{x}_0}\!\left({x}_0,{u}_0\right)\, \delta{x}\ +\ \frac{\partial f_{\dot\gamma}}{\partial{u}_0}\!\left({x}_0,{u}_0\right)\, \delta{u} \right],
\end{equation*}
\begin{equation*}
\label{Gleich: Perturbation_Chi}
\delta \dot{\chi} = w_\chi \left[ \frac{\partial f_{\dot\chi}}{\partial{x}_0}\!\left({x}_0,{u}_0\right)\, \delta{x}\ +\ \frac{\partial f_{\dot\chi}}{\partial{u}_0}\!\left({x}_0,{u}_0\right)\, \delta{u} \right],
\end{equation*}
where $w_\gamma=1/mV_0$ and $w_\chi=1/(mV_0\cos(\gamma_0))$. During the hypersonic endoatmospheric glide phase, \(V_0\) is large such that
\begin{equation*}
m V_0 \gg 
\left|\ \frac{\partial f_{\dot\gamma}}{\partial{x}_0}\left({x}_0,{u}_0\right)\, \delta{x}\ +\ \frac{\partial f_{\dot\gamma}}{\partial{u}_0}\left({x}_0,{u}_0\right)\, \delta{u}\ \right| ,
\end{equation*}
\begin{equation*}
m V_0 \gg 
\left|\ \frac{\partial f_{\dot\chi}}{\partial{x}_0}\left({x}_0,{u}_0\right)\, \delta{x}\ +\ \frac{\partial f_{\dot\chi}}{\partial{u}_0}\left({x}_0,{u}_0\right)\, \delta{u}\ \right| .
\end{equation*}
Hence, the coupling between ${T}_2(\mu,\gamma,\chi)$ and $(\dot{\gamma}, \dot{\chi})^\top$ in \eqref{eqn:Advanced Rotational Kinematics 1} can be neglected for the considered operating range, leading to the following simplified, state–affine relation:
\begin{equation}
\label{eqn:Lineare NDI-Außenschleife 1}
{x}_1 =
\begin{bmatrix} 
 \dot{\mu} \\ 
 \dot{\alpha} \\ 
 \dot{\beta}
\end{bmatrix} = 
{T}_1 (\alpha,\beta) 
\begin{bmatrix} 
 p \\ 
 q \\ 
 r 
\end{bmatrix}.
\end{equation}

In many applications for fast flying vehicles, the contribution of the gyroscopic cross product in \eqref{eqn:Moments_EOM} can be assumed to be small relative to the aerodynamic moments. Neglecting gravity–induced, centrifugal, and Coriolis \emph{moments} for the attitude channel then yields
\begin{equation}
\label{eqn:Surrogate Moments}
{x}_2 =
\begin{bmatrix} 
 \dot{p} \\ 
 \dot{q} \\ 
 \dot{r} 
\end{bmatrix} = {I}^{-1}
\begin{bmatrix} 
M_x \\ 
M_y \\ 
M_z 
\end{bmatrix},
\end{equation}
with aerodynamic moments modeled as
\begin{equation}
\begin{bmatrix} 
M_x \\ 
M_y \\ 
M_z 
\end{bmatrix}
= \underbrace{\bar{q}\, S\, l}_{\Gamma}
\begin{bmatrix}
 C_{M_x}({\sigma}_0) \\ 
 C_{M_y}({\sigma}_0) \\ 
 C_{M_z}({\sigma}_0) 
\end{bmatrix},
\end{equation}
where \(\bar{q}\) is the dynamic pressure, \(S\) the aerodynamic reference area, and \(l\) the reference length.

To obtain a linear model scheduled by \({\sigma}_0\), we approximate the moment–coefficient increments affinely via state perturbations, obtaining a Jacobian depending on \({\sigma}_0\):
\begin{equation}\label{Gl:affine_model_of_moment_coefs}
\begin{bmatrix}
 C_{M_x}({\sigma}_0)  \\ 
 C_{M_y} ({\sigma}_0)  \\ 
 C_{M_z} ({\sigma}_0) 
\end{bmatrix} = {C}_J({\sigma}_0) 
\begin{bmatrix}
\mu \\ 
\alpha \\ 
\beta \\ 
 p \\ 
 q \\ 
 r
\end{bmatrix},
\end{equation}
where
\begin{equation*}
{C}_J({\sigma}_0) =
\begin{bmatrix*}
 C_{M_x,\mu} ({\sigma}_0) & C_{M_y,\mu} ({\sigma}_0) & C_{M_z,\mu} ({\sigma}_0)\\
 C_{M_x,\alpha} ({\sigma}_0) & C_{M_y,\alpha} ({\sigma}_0) & C_{M_z,\alpha} ({\sigma}_0)\\
 C_{M_x,\beta} ({\sigma}_0) & C_{M_y,\beta} ({\sigma}_0) & C_{M_z,\beta} ({\sigma}_0)\\
 C_{M_x,p} ({\sigma}_0) & C_{M_y,p} ({\sigma}_0) & C_{M_z,p} ({\sigma}_0)\\
 C_{M_x,q} ({\sigma}_0) & C_{M_y,q} ({\sigma}_0) & C_{M_z,q} ({\sigma}_0)\\
 C_{M_x,r} ({\sigma}_0) & C_{M_y,r} ({\sigma}_0) & C_{M_z,r} ({\sigma}_0)
\end{bmatrix*}^T.
\end{equation*}

Combining Eqs.~\eqref{eqn:Lineare NDI-Außenschleife 1}, \eqref{eqn:Surrogate Moments}, and \eqref{Gl:affine_model_of_moment_coefs} delivers the following linear representation for the attitude dynamics:
\begin{equation}
\label{Gl: Surrogate LPV Model}
\dot{{x}}
={A} ({\sigma}_0)
\underbrace{\begin{bmatrix*} 
 \mu \\ 
 \alpha\\
 \beta\\
 p\\
 q\\
 r
\end{bmatrix*}}_{{x}} 
+ {B}({\sigma}_0) 
\underbrace{\begin{bmatrix*} 
 M_x\\
 M_y\\
 M_z
\end{bmatrix*}}_{{u}},
\end{equation}
with
\begin{equation*}
{A} ({\sigma}_0) =
\begin{bmatrix*}
{0}_{3 \times 3} \hspace{0.3cm} {T}_1 (\alpha,\beta) \\ 
{\Xi}({\sigma}_0)
\end{bmatrix*},
\qquad
{B} ({\sigma}_0) =
\begin{bmatrix*} 
0_{3\times 3}\\ 
{I}^{-1}
\end{bmatrix*},
\end{equation*}
and
\begin{equation*}
{\Xi}({\sigma}_0) = {I}^{-1}\, \Gamma\, {C}_J({\sigma}_0).
\end{equation*}



\end{document}